\documentclass{irmaart}
\input{ha.sty}
\usepackage{gastex}

\usepackage[displaymath]{lineno}
		 \nolinenumbers 

\usepackage[]{xr-hyper}
\usepackage[hypertex,hyperindex,pagebackref,final]{hyperref}

\hypersetup{
		colorlinks = true,
		linkcolor = red,
		anchorcolor = black,
		citecolor = blue, 
		filecolor = cyan,
		menucolor = red,
		runcolor = cyan,
		urlcolor = magenta,
		linkbordercolor = {white},
		linktocpage = true
}

\newcommand\cqfd{\skip10=\parfillskip\parfillskip=0pt
\enspace\hfill\symbolecqfd\par\parfillskip=\skip10\par\medskip}
\newcommand\symbolecqfd{\rlap{$\sqcap$}$\sqcup$}

\newcommand\mathcalC{\mathcal{C}}

\newcommand\mathcalK{\mathcal{K}}
\newcommand\mathcalL{\mathcal{L}}

\newcommand\mathcalV{\mathcal{V}}

\newcommand\psvA{\textbf{Ap}}
\newcommand\Com{\textbf{Com}}
\newcommand\psvD{\textbf{D}}
\newcommand\psvDA{\textbf{DA}}

\newcommand\psvG{\textbf{G}}
\newcommand\psvI{\textbf{I}}
\newcommand\psvJ{\textbf{J}}
\newcommand\psvK{\textbf{K}}
\newcommand\LL{\textbf{L}}
\newcommand\RR{\textbf{R}}
\newcommand\Sl{\textbf{Sl}}
\newcommand\psvV{\textbf{V}}
\newcommand\psvW{\textbf{W}}
\newcommand\psvX{\textbf{X}}
\renewcommand\phi\varphi
\newcommand\calC{\mathcal{C}}
\newcommand\calJ{\mathcal{J}}

\newcommand\calL{\mathcal{L}}
\newcommand\calR{\mathcal{R}}
\newcommand\calV{\mathcal{V}}
\newcommand\calW{\mathcal{W}}
\newcommand\inv{^{-1}}

\font\petite=cmmi10 at 8pt
\newcommand\malcev{\mathbin{\hbox{$\bigcirc$\rlap{\kern-9pt\raise0,75pt\hbox{\petite m}}}}}

\newcommand\llbrack{[\![}
\newcommand\rrbrack{]\!]}

\newcommand{\synt}{\operatorname{Synt}}
\newcommand{\Det}{\operatorname{Det}}
\newcommand{\coDet}{\operatorname{coDet}}
\newcommand{\Pol}{\operatorname{Pol}}
\newcommand{\UPol}{\operatorname{UPol}}
\newcommand{\FO}{\operatorname{FO}}

\newcommand\dast{\mathbin{\ast\!\ast}}

\begin{document}
\markboth{H. Straubing, P. Weil}{Varieties}
\title{Varieties}
\author{Howard Straubing$^{1}$\thanks{Work partially supported by NSF Grant
    CCF-0915065}  and Pascal Weil$^{2}$\thanks{Work partially supported by ANR Grant ANR-16-CE40-0007 (project \textsc{DeLTA}) and by ReLaX, CNRS UMI 2000.}}
\address{ $^1$Computer Science Department\\ Boston College\\  Chestnut Hill, Massachusetts 02467, USA\\
  email:\,\url{straubin@cs.bc.edu}
  \\[4pt]
  $^2$Univ. Bordeaux, CNRS, LaBRI, UMR 5800, F-33400 Talence, France\\
  email:\,\url{pascal.weil@labri.fr} }


\maketitle\label{chapter16}

\setcounter{page}{509}


This chapter is devoted to the theory of varieties, which provides an
important tool, based in universal algebra, for the classification of
regular languages.  In the introductory section, we present a number
of examples that illustrate and motivate the fundamental concepts.  We
do this for the most part without proofs, and often without precise
definitions, leaving these to the formal development of the theory
that begins in Section 2.  Our presentation of the theory draws heavily on the work of Gehrke, Grigorieff and Pin~\cite{Gehrke&Grigorieff&Pin:2008} on the equational theory of lattices of regular languages. In the subsequent sections we consider in more detail aspects of varieties that were only briefly evoked in the introduction:  Decidability, operations on languages, and characterizations in formal logic.

\section{Motivation and examples}

We refer the readers to Chapter~\ref{Chap1}, and specifically to Sections 4.2 and 4.3 of that chapter, for the notion of a language recognized by a morphism into a finite monoid, and for the definition of the syntactic monoid $\synt(L)$ of a language $L$.

\subsection{Idempotent and commutative monoids}\label{SW:subsection:j_1}

When one begins the study of abstract algebra, groups are usually
encountered before semigroups and monoids.  The simplest example of a
monoid that is not a group is the set $\{0,1\}$ with the usual
multiplication.  We denote this monoid $U_1$\index{U@$U_1$}.

What are the regular languages recognized by $U_1$?  If $A$ is a
finite alphabet and $\phi\colon A^*\to U_1$ is a morphism, then any
language $L\subseteq A^*$ recognized by $\phi$---that is, any set of
the form $\phi^{-1}(X)$ where $X\subseteq U_1$---has either the form
$B^*$ or $A^*\setminus B^*$, where $B\subseteq A$. In particular,
membership of a word $w$ in $L$ depends only on the set $\alpha(w)$ of
letters occurring in $w$ (see Example 4.9 in Chapter~\ref{Chap1}).

The property `membership of $w$ in $L$ depends only on $\alpha(w)$'
is preserved under union and complement, and thus defines a boolean
algebra of regular languages.  Of course, not every language in this
boolean algebra is recognized by $U_1$; for example, we could take
$L=a^*\cup b^*$. However, it follows from basic properties of the
syntactic monoid
that this boolean algebra consists of precisely the languages recognized by
finite direct products of copies of $U_1$.

We have thus characterized a syntactic property of regular languages
in terms of an algebraic property of its syntactic monoid.  The family
of finite monoids that divide a direct product of a finite number of
copies of $U_1$ is itself closed under finite direct products and
division.  Such a family of finite monoids is called a \emph{pseudovariety}\index{pseudovariety}. This particular pseudovariety is often denoted $\psvJ_1$\index{J@$\psvJ_1$} in the literature\footnote{It is also written $\Sl$ because its elements are called semilattices\index{semilattice}.}.

\subsubsection{Decidability and equational description}

Thus if we
want to decide whether a given language $L\subseteq A^*$ has this
syntactic property, we can compute $\synt(L)$ and try to determine whether
$\synt(L)\in\psvJ_1$. But how do we do \emph{that?} There are, after all,
infinitely many monoids in $\psvJ_1$. We can, however, bound the
size of the search space in terms of $|A|$. It is not hard to prove
that if $M$ is a finite monoid, and
\begin{displaymath}
\phi\colon A^*\to \underbrace{M\times\cdots\times M}_{r\quad{\rm times}}
\end{displaymath}
is  a morphism,  then $N=\phi(A^*)$ embeds into  
\begin{displaymath}
\underbrace{M\times\cdots\times M}_{s\quad{\rm times}},
\end{displaymath}
where $s=|M|^{|A|}$. This settles, in a not very satisfactory way, the
question of deciding whether $\synt(L)$ is in $\psvJ_1$: The resulting
`decision procedure'---check all the divisors of $U_1^{2^{|A|}}$ and
see if $\synt(L)$ is isomorphic to any of them!---is of course
ridiculously impractical.  Fortunately, there is a better approach:
$U_1$ is both commutative and idempotent (\emph{i.e.,} all its elements
are idempotents).These two properties are preserved under direct
products and division, and consequently shared by all members of $\psvJ_1$. That is, the idempotent and commutative monoids\index{monoid!idempotent and commutative} form a
pseudovariety that contains $\psvJ_1$. Conversely, every idempotent
and commutative finite monoid belongs to $\psvJ_1$. To see this, we
make note of a fact that will play a large role in this chapter: If
$M$ is a finite monoid and $\phi\colon A^*\to M$ an onto morphism, then
\begin{displaymath}
M\prec \prod_{m\in M} \synt(\phi^{-1}(m)).
\end{displaymath}
In particular, \emph{every pseudovariety is generated by the syntactic
monoids it contains.} We now observe that if
$\alpha(w_1)=\alpha(w_2)$, and if $\phi\colon A^*\to M$ is a morphism
onto an idempotent and commutative monoid, then $\phi(w_1)=\phi(w_2)$,
since we can permute letters and eliminate duplications in any word
$w$ without changing its value under $\phi$. Thus each
$\phi^{-1}(m)$ satisfies our syntactic property, and so by the remark just made,
$M\in \psvJ_1$.

We can express `$M$ is idempotent and commutative' by saying that $M$
\emph{satisfies the identities}\index{identity} $xy=yx$ and $x^2=x$. This means that
these equations hold no matter how we substitute elements of $M$ for
the variables $x$ and $y$. This equational characterization of $\psvJ_1$ provides a much more satisfactory procedure for determining if a
monoid $M$ belongs to $\psvJ_1:$ If $M$ is given by its
multiplication table, then we can verify the identities in time
polynomial in $|M|$.

\subsubsection{Connection to logic}

Before leaving this example, we
note a connection with formal logic.  We express properties of words
over $A^*$ by sentences of first-order logic in which variables denote
positions in a word.  For each $a\in A$, our logic contains a unary
predicate $Q_a$, where $Q_ax$ is interpreted to mean `the letter in
position $x$ is $a$'.  We allow only these formulas $Q_ax$ as atomic
formulas---in particular, we do not include equality as a predicate.
A sentence in this logic, for example (with $A=\{a,b,c\}$)
\begin{displaymath}
\exists x\exists y\forall z(Q_ax\wedge Q_by\wedge \neg Q_cz)
\end{displaymath}
defines a language over $A^*$, in this case the set of all words
containing both $a$ and $b$, but with no occurrence of $c$. It is easy
to see that the languages definable in this logic are exactly those in
which membership of a word $w$ depends only on $\alpha(w)$.

The following theorem summarizes the results of this subsection.

\begin{theorem}
    Let $A$ be a finite alphabet and let $L\subseteq A^*$ be a regular
    language.  The following are equivalent.
    \begin{enumerate}
	\item[(i)] Membership of $w$ in $L$ depends only on the set
	$\alpha(w)$ of letters appearing in $w$.
	\item[(ii)] $\synt(L)\in\psvJ_1$, that is, $\synt(L)$ divides a finite
	direct product of copies of $U_1$.
	\item[(iii)] $\synt(L)$ satisfies the identities $xy=yx$ and $x^2=x$.
	\item[(iv)] $L$ is definable by a first-order sentence over the
	predicates $Q_a$, $a\in A$.
    \end{enumerate}
\end{theorem}

\subsection{Piecewise-testable languages}\label{SW:subsection:j}

Suppose that instead of testing for occurrences of individual letters
in a word, we test for occurrences of non-contiguous sequences of
letters, or \emph{subwords}\index{subword}.  More precisely, we say that $v=a_1\cdots
a_k$, where each $a_i\in A$, is a subword of $w\in A^*$ if
\begin{displaymath}
w=w_0a_1w_1\cdots a_kw_k
\end{displaymath}
for some $w_0,\ldots, w_k\in A^*$. We also say that the empty word 1
is a subword of every word in $A^*$. The set of all words in $A^*$
that contain $v$ as a subword is thus the regular language
\begin{displaymath}
L_v=A^*a_1A^*\cdots a_kA^*.
\end{displaymath}
We say that a language is \emph{piecewise-testable}\index{language!piecewise-testable} if it belongs to
the boolean algebra generated by the $L_v$.

\subsubsection{Decidability and equational description}

It is not clear that we can effectively decide whether a given regular
language is piecewise testable.  For the language class of
\ref{SW:subsection:j_1}, we were able to settle this question by in
effect observing that for every finite alphabet $A$ there were only
finitely many languages of the class in $A^*$. For piecewise-testable
languages, this is no longer the case.  It is possible,
however, to obtain an algebraic characterization of the
piecewise-testable languages, and this leads to a fairly efficient
decision procedure.  We first note two relatively easy-to-prove facts.
First, the monoids $\synt(L_v)$ are all \emph{${\mathcal J}$-trivial}\index{J-trivial@$\mathcal{J}$-trivial}: This
means that if $m,m',s,t,s',t'\in \synt(L_v)$ are such that $m=s'm't'$,
$m'=smt$, then $m=m'$. Second, the family $\psvJ$\index{J@$\psvJ$} of ${\mathcal
J}$-trivial monoids forms a pseudovariety.  It follows then that the
syntactic monoid of every piecewise-testable language is ${\mathcal
J}$-trivial.  A deep theorem, due to I. Simon~\cite{Simon:1975}, shows that
the converse is true as well: Every language recognized by a finite
${\mathcal J}$-trivial monoid is piecewise-testable.
 
Clearly, we can effectively determine, from the multiplication table
of a finite monoid $M$, all the pairs $(m,m')\in M\times M$ such that
$m'=smt$ for some $s,t\in M$, and thus determine if $M\in\psvJ$.
This gives us an algebraic decision procedure for
piecewise-testability.
 
Can the pseudovariety $\psvJ$ be defined by identities in the same
manner as $\psvJ_1$?  The short answer is `no'.  This is because
satisfaction of an identity $u=v$, where $u$ and $v$ are words over an
alphabet $\{x,y,\ldots\}$ of variables, is preserved by \emph{infinite}
direct products as well as finite direct products and divisors.
Consider now the monoids
\begin{displaymath}
M_j=\{1,m,m^2,\ldots,m^j=m^{j+1}\}.
\end{displaymath}
Each $M_j\in \psvJ$, but $\prod_{j\geq 1} M_j$ contains an
isomorphic copy of the infinite cyclic monoid $\{1,a,a^2,\ldots\}$,
which has every finite cyclic group as a quotient.  Thus every
identity satisfied by all the monoids in $\psvJ$ is also satisfied by
all the finite cyclic groups, which are not in $\psvJ$.
 
In spite of this, we can still obtain an equational description of
$\psvJ$, provided we adopt an expanded notion of what constitutes an
identity.  If $s$ is an element of a finite monoid $M$, then we denote
by $s^{\omega}$ the unique idempotent power of $s$. We will allow identities in
which the operation $x\mapsto x^{\omega}$ is allowed to appear; these are special instances of what we will call \emph{profinite identities.}  It is not hard to see that satisfaction of these new identities is preserved under finite direct products and quotients, and thus every
set of such identities defines a pseudovariety.

For example,
the profinite identity
\begin{displaymath}
x^{\omega}=xx^{\omega}
\end{displaymath}
is satisfied by precisely the finite monoids that contain no nontrivial groups.  This is the pseudovariety of \emph{aperiodic monoids}\index{monoid!aperiodic}, which we denote $\psvA$\index{A@$\psvA$}.  Similarly, the profinite identity
\begin{displaymath}
x^{\omega}=1
\end{displaymath}
defines the pseudovariety $\psvG$ of finite groups.  As was the case with $\psvJ$, neither of these pseudovarieties can be defined by a set of ordinary identities.

 It can be shown
that the pseudovariety $\psvJ$ of finite ${\mathcal J}$-trivial monoids is
defined by the pair of profinite identities
\begin{align}
(xy)^{\omega}x&=(xy)^{\omega}\nonumber \\
y(xy)^{\omega}&=(xy)^{\omega}\nonumber,
\end{align}
or, alternatively, by the pair
\begin{align}
(xy)^{\omega}&=(yx)^{\omega}\nonumber \\
xx^{\omega}&=x^{\omega}\nonumber .
\end{align}

\subsubsection{Connection with logic}\label{SW:sssec: sigma1}

Let us supplement the
first-order logic for words that we introduced earlier with atomic
formulas of the form $x<y$, which is interpreted to mean `position $x$
is strictly to the left of position $y$'.  The language $L_v$, where
$v=a_1\cdots a_k$, is defined by the sentence
\begin{displaymath}
\exists x_1\exists x_2\cdots \exists x_k(x_1<x_2\wedge
x_2<x_3\wedge\cdots\wedge x_{k-1}<x_k\wedge
Q_{a_1}x_1\wedge\cdots\wedge Q_{a_k}x_k).
\end{displaymath}
This is a $\Sigma_1$-sentence\index{Sigma@$\Sigma_1$-sentence}---one in which all the quantifiers are
in a single block of existential quantifiers at the start of the
sentence.  It follows easily that a language is piecewise-testable if
and only if it is defined by a boolean combination of
$\Sigma_1$-sentences.
 
The following theorem summarizes the results of this subsection.
 
\begin{theorem}
    Let $A$ be a finite alphabet and let $L\subseteq A^*$ be a regular
    language.  The following are equivalent.
    \begin{enumerate}
	\item[(i)]  $L$ is piecewise testable.
	\item[(ii)] $\synt(L)\in\psvJ$, that is, $\synt(L)$ is ${\mathcal J}$-trivial.
	\item[(iii)] $\synt(L)$ satisfies the identities $(xy)^{\omega}=(yx)^{\omega}$
	and $xx^{\omega}=x^{\omega}$.
	\item[(iv)] $\synt(L)$ satisfies the identities $(xy)^{\omega}x=y(xy)^{\omega}$.
	\item[(v)] $L$ is definable by a boolean combination of
	$\Sigma_1$-sentences over the predicates $<$ and $Q_a$, $a\in A$.
    \end{enumerate}
\end{theorem}

\subsection{Pseudovarieties of monoids and varieties of languages}\label{SW:sec: varieties}

We tentatively extract a few general principles from the preceding
discussion.  These will be explored at length in the subsequent
sections.  Given a pseudovariety $\psvV$ of finite monoids and a
finite alphabet $A$, we form the family $A^*{\mathcal V}$ of all regular
languages $L\subseteq A^*$ for which $\synt(L)\in \psvV$. We can think
of ${\mathcal V}$ itself as an operator that associates to each finite
alphabet $A$ a family of regular languages over $A$. ${\mathcal V}$ is
called a \emph{variety of languages}\index{variety of languages}.  (We will give a very different, although equivalent definition of this term in our formal discussion in Section 2.) From our earlier observation
that pseudovarieties are generated by the syntactic monoids they
contain, it follows that if $\psvV$ and $\psvW$ are distinct
pseudovarieties, then the associated varieties of languages ${\mathcal V}$
and ${\mathcal W}$ are also distinct.  Thus \emph{there is a one-to-one
correspondence between varieties of languages and pseudovarieties of
finite monoids.}
 
Often we are interested in the following sort of decision problem:
Given a regular language $L\subseteq A^*$, does it belong to some
predefined family ${\mathcal V}$ of regular languages, for example, the
languages definable in some logic?  If ${\mathcal V}$ forms a variety of
languages, then we can answer the question if we have some effective
criterion for determining if a given finite monoid belongs to the
corresponding pseudovariety $\psvV$.  (The converse is true as well: if
we could decide the question about membership in the variety of
languages, we would be able to decide membership in $\psvV$.)

\emph{Pseudovarieties are precisely the families of finite monoids
defined by sets of profinite identities.} For the time being this
assertion---a theorem due to Reiterman--- will have to remain somewhat
vague, since we haven't even come close to saying what a
profinite identity actually is!  Such equational characterizations of
pseudovarieties are frequently the source of the decision procedures
discussed above.

If ${\mathcal V}$ is a variety of languages, then, as we have seen, each
$A^*{\mathcal V}$ is closed under boolean operations.  Observe further
that if $L\in A^*{\mathcal V}$ and $v\in A^*$, then both of the \emph{quotient}\index{language!quotient} languages
\begin{align}
v^{-1}L&=\{w\in A^* \mid vw\in L\}\nonumber\\
Lv^{-1}&=\{w\in A^*\mid wv\in L\}\nonumber
\end{align}
are in $A^*{\mathcal V}$, because any monoid recognizing $L$ also
recognizes the quotients.  For the same reason, if $\phi\colon B^*\to A^*$ is a
morphism, $\phi^{-1}(L)$ is in $B^*{\mathcal V}$. An important result,
due to Eilenberg, showed that these closure properties characterize
varieties of languages.

\begin{theorem}\label{SW:thm:eilenberg}
    Let ${\mathcal V}$ assign to each finite alphabet $A$ a family
    $A^*{\mathcal V}$ of regular languages in $A^*$.  ${\mathcal V}$ is a
    variety of languages if and only if the following three conditions
    hold:
    \begin{enumerate}
	\item[(i)] Each $A^*{\mathcal V}$ is closed under boolean operations.
	
	\item[(ii)] If $L\in A^*{\mathcal V}$ and $w\in A$, then $w^{-1}L\in
	A^*{\mathcal V}$, and $Lw^{-1}\in A^*{\mathcal V}$.

	\item[(iii)] If $L\in A^*{\mathcal V}$ and $\phi\colon B^*\to A^*$ is a
	morphism of finitely generated free monoids, then
	$\phi^{-1}(L)\in B^*{\mathcal V}$.
    \end{enumerate}
\end{theorem}

This theorem can be quite useful for showing, in the absence of an
explicit algebraic characterization of the corresponding pseudovariety
of monoids, that a combinatorially or logically defined family of
languages forms a variety.  We conclude from this that such an algebraic
characterization in principle exists.

Although it is somewhat involved, Theorem~\ref{SW:thm:eilenberg} is quite elementary, see \cite{Eilenberg:1976,Pin:1986}. In the next section we will revisit the definition of varieties of languages and profinite identities in a way that will permit
us to prove both Theorem~\ref{SW:thm:eilenberg} and Reiterman's theorem in a single argument.

Before we proceed with this program, we briefly describe certain
classes of regular languages which admit syntactic characterizations
(that is: characterizations in terms of syntactic monoids and
syntactic morphisms), but which are not varieties in the sense described above.

\subsection{Extensions}\label{SW:sec: extensions}

Interesting classes of regular languages frequently admit
characterizations in terms of their syntactic monoids and syntactic
morphisms, and the theory sketched above is meant to provide a formal
setting for this algebraic classification of regular languages.
However, the framework is not adequate to capture all the examples of
interest that arise.  Here we give three examples.

Consider, first, the family $A^*{\mathcal K}_1$ of languages $L\subseteq
A^*$ for which membership of $w$ in $L$ is determined by the leftmost letter
of $w$. This class forms a boolean algebra closed under quotients, but
is not a variety of languages.  To see this, note that $a(a+b)^*\in\{a,b\}^*{\mathcal K}_1$ and $c^*a(a+b+c)^*\notin\{a,b,c\}^*{\mathcal K}_1$, even though the two languages have the same
syntactic monoid.  Alternatively, we can reason using
Theorem~\ref{SW:thm:eilenberg}, and note that the second language is an
inverse homomorphic image of the first, and thus ${\mathcal K}_1$ fails to
be a variety of languages. More generally, we can define the family 
$A^*{\mathcal K}_d$ of languages $L$ for which membership of $w$ in $L$ depends only on the leftmost $\min(|w|,d)$ letters of $w$, as well as
$A^*{\mathcal K}=\bigcup_{d>0}A^*{\mathcal K}_d$.  All these families are closed under boolean operations and quotients, yet fail to be varieties of languages.

We obtain an example with a similar flavor if we supplement the
predicate logic described earlier by atomic formulas $x\equiv_q 0$,
where $q>1$, which is interpreted to mean that position $x$ is
divisible by $q$. (We assume that positions in a word are numbered,
beginning with $1$ for the leftmost position.)  We denote by $A^*\mathcal{QA}$\index{QA@$\mathcal{QA}$} the family of languages over $A^*$ definable in this logic.
Languages in $A^*\mathcal{QA}$ arise as the regular languages definable in
the circuit complexity class $AC^0$\index{AC0@$AC^0$} (see \cite{Barrington&Compton&Straubing&Therien:1992}). Each $A^*\mathcal{QA}$ is a boolean
algebra closed under quotients, however $\mathcal{QA}$ is not a variety of
languages: To see this, consider the morphism $\{a,b\}^*\to\{a\}^*$ that maps
$a$ to $a$ and $b$ to the empty string.  The set
$\{a^{2n} \mid n\geq 0\}$ is in $\{a\}^*\mathcal{QA}$, as it is defined by by
the sentence
\begin{displaymath}
\forall x(\forall y(y\leq x)\rightarrow x\equiv_2 0).
\end{displaymath}
However the inverse image of this language under the morphism is
the set of strings over $\{a,b\}$ with an even number of occurrences
of $a$, and it is possible to prove by model-theoretic means that this
language is not definable in our logic.

Finally, consider the family $A^*{\mathcal J}^+$\index{J+@$\mathcal{J}^+$} of languages definable by
$\Sigma_1$-sentences over the predicates $<$ and $Q_a$ with $a\in A$
(in contrast to the languages definable by boolean combinations of
$\Sigma_1$-sentences, which we considered earlier).  It is easy to see
that if $L\in A^*{\mathcal J}^+$ and $w\in L$, then $L_w\subseteq L$. This
readily implies that $A^*{\mathcal J}^+$ is not closed under complement,
since, for example, the complement of $(a+b)^*a(a+b)^*$ does not have
this property.  Thus ${\mathcal J}^+$ is not a variety of languages.  On
the other hand, it does satisfy many of the properties of varieties of
languages: It is closed under finite unions and intersections,
quotients, and inverse images of morphisms between free monoids.

It turns out that each of these three examples admits an algebraic
characterization in terms of classes that are very much like
pseudovarieties.  For our first example, in which membership of a word
in a language is determined by the leftmost letter, the correct
generalization of pseudovarieties was already known to Eilenberg: One
looks not at the syntactic monoid of a language $L$, but at the image
of the set $A^+$ of nonempty words under the syntactic morphism.  This
is called the \emph{syntactic semigroup}\index{semigroup!syntactic} of $L$. We can define
\emph{pseudovarieties of finite semigroups}\index{pseudovariety!of finite semigroups} just as we defined
pseudovarieties of finite monoids.  Then $L\in A^*{\mathcal K}_1$ if and
only if its syntactic semigroup belongs to the pseudovariety of
semigroups defined by the identity $xy=x$. While ${\mathcal K}_1$ is not
closed under inverse images of morphisms between free monoids, it is
closed if we restrict ourselves to \emph{non-erasing} morphisms\index{morphism!non-erasing}---those that map
every letter to a nonempty word.

We can use a similar method to characterize the class $\mathcal{QA}$. Once
again we look not just at the syntactic monoid of a language $L$, but
at the additional structure provided by the syntactic morphism
$\eta_L$. It is known that $L\in A^*\mathcal{QA}$ if and only if for every
$k\geq 0$, $\eta_L(A^k)$ contains no nontrivial groups \cite{Barrington&Compton&Straubing&Therien:1992}.  The family
${\bf QA}$\index{QA@{\bf QA}} of morphisms from free monoids onto finite monoids with
this property forms a kind of pseudovariety with respect to
appropriately modified definitions of direct product and division.  An
equational characterization of {\bf QA} is provided by the identity
\begin{displaymath}
(x^{\omega-1}y)^{\omega}=(x^{\omega-1}y)^{\omega+1},
\end{displaymath}
where the identity is interpreted in the following sense: $\phi\in
{\bf QA}$ if and only if for all words $u$ and $v$ \emph{of the same
length,} $x=\phi(u)$ and $y=\phi(v)$ satisfy the identity.  $\mathcal{QA}$
is closed under inverse images of morphisms $f\colon B^*\to A^*$ such that
$f(B)\subseteq A^k$ for some $k\geq 0$; these are called \emph{length
multiplying morphisms}\index{morphism!length multiplying}. In fact, these last two examples are
instances of a single phenomenon: Families of morphisms
$\phi\colon A^*\to M$ onto finite monoids that form pseudovarieties with
respect to some underlying composition-closed class ${\mathcal C}$ of
morphisms between free monoids.

For the example ${\mathcal J}^+$ of $\Sigma_1$-definable languages, the
algebraic characterization involves a different generalization of
pseudovarieties.  Here the additional structure on the syntactic
monoid is provided by the embedding of $\eta_L(L)$ in $\synt(L):$ If
$m_1,m_2\in M$ then we say $m_1\leq_L m_2$ if
\begin{multline}
\{(s,t)\in \synt(L)\times \synt(L) \mid sm_2t\in \eta_L(L)\} \\
\subseteq \{(s,t)\in \synt(L)\times \synt(L) \mid sm_1t\in \eta_L(L)\}.\nonumber
\end{multline}
This gives a partial order on $\synt(L)$ compatible with multiplication (see Section 4.4 in  Chapter~\ref{Chap1}).
We then find that $L\in A^*{\mathcal J}^+$ if and only if this ordered
syntactic monoid satisfies the \emph{inequality} $x\leq 1$ for each element $x$. The family
of partially-ordered monoids satisfying this inequality is a
pseudovariety of ordered finite monoids---it is closed under finite
direct products, and order-compatible submonoids and quotients.  The theory of pseudovarieties of ordered monoids and the corresponding \emph{positive varieties}\index{variety!positive} of languages is due to Pin~\cite{Pin:1995}

In the next section we will formally develop the framework that gives
the correspondence between pseudovarieties and language varieties, and
the definition by profinite identities, in a very general setting.
Pseudovarieties of finite monoids, as well as all the generalizations
mentioned above, will appear as special cases.

\section{Equations, identities and families of languages}\label{SW:sec: eqs}

The original statement of Eilenberg's theorem dealt
exclusively with varieties of languages. Here we will show how to use a whole hierarchy
of increasingly complex equational characterizations of increasingly
structured families of languages. 
Before we describe these results, we need to give a quick 
introduction to the free profinite monoid and its connection to the 
theory of regular languages

\subsection{The free profinite monoid}\label{SW:sec: profinite monoid}

Say that a finite monoid $M$ \emph{separates} two words $u,v\in A^*$
if there exists a morphism $\phi\colon A^*\to M$ such that $\phi(u)
\ne \phi(v)$.  Note that if $u\ne v$, there always exists such a
monoid.  Indeed, for each $n\ge 1$, consider the quotient monoid $A^*/
A^{\ge n}$: it consists of the set of words of length less than $n$,
plus a zero, and each product with length at least $n$ (in $A^*$) is
equal to 0.  Then $A^*/ A^{\ge n}$ separates $u$ and $v$ if $n >
\max(|u|,|v|)$. We denote by $r(u,v)$ the minimum cardinality of a 
monoid separating $u$ and $v$.

The \emph{profinite distance}\index{distance!profinite}\index{profinite!distance} on $A^*$ is defined by letting $d(u,v)
= 2^{-r(u,v)}$ if $u\ne v$ and $d(u,u) = 0$.  One verifies easily that
$d$ is in fact an ultrametric distance (it satisfies the ultrametric
inequality $d(u,v) \le \max(d(u,w),d(v,w))$, stronger than the
triangle inequality), and the above discussion shows that the
resulting metric space is Hausdorff.

The topology thus defined on $A^*$ is not especially interesting: we get a discrete space,
where a sequence $(u_n)_n$ converges to a word $u$ if and only if
$(u_n)_n$ is ultimately equal to $u$\dots This can be verified using
the monoids $A^*/ A^{\ge n}$ described above.  There are, however,
non-trivial Cauchy sequences.  In fact, one can show the following.

\begin{proposition}
    A sequence $(u_n)_n$ is Cauchy if and only if, for each morphism 
    $\phi\colon A^*\to M$ into a finite monoid, the sequence 
    $(\phi(u_n))_n$ is ultimately constant.
\end{proposition}
For instance, if $u$ is a word, then $(u^{n!})_n$ is a Cauchy sequence
(this can be deduced from the fact that its image under any morphism into
a finite monoid is ultimately constant), but it is non-trivial if
$u\ne 1$.  In topological terms, the uniform structure defined by the
profinite distance is non-trivial.

Using a classical construction from topology (analogous to the
construction of the real numbers from the rationals), we can now
consider the completion of $(A^*,d)$, denoted by $\widehat{A^*}$.  It
can be viewed as the quotient of the set of Cauchy sequences in $(A^*,d)$ by the relation identifying two sequences $(u_n)$ and $(v_n)$ if the
mixed sequence, alternating the terms of $(u_n)$ and $(v_n)$, is
Cauchy as well.  In particular, $A^*$ is naturally seen as a dense
subset of $\widehat{A^*}$.

The following results can be verified by elementary means.

\begin{proposition}\label{SW:elementary profinite}
    Let $A$ be an alphabet.
    \begin{conditions}
   \item  The multiplication operation $(u,v) \mapsto uv$ in
    $A^*$ is uniformly continuous.

    \item  Every morphism $\phi\colon A^* \to B^*$ between free monoids,
    and every morphism $\psi\colon A^* \to M$ from a free monoid to a
    finite monoid (equipped with the discrete distance) is uniformly
    continuous.
    
    \item $\widehat{A^*}$ is a compact space.
    \end{conditions}
\end{proposition}

By a standard property of completions, it follows from
Proposition~\ref{SW:elementary profinite} (1) that the multiplication of
$A^*$ can be extended to $\widehat{A^*}$: the resulting monoid is
called the \emph{free profinite monoid on $A$}\index{profinite!monoid}\index{monoid!free profinite}.  Similarly,
Proposition~\ref{SW:elementary profinite} (3) shows that each morphism
$\phi\colon A^* \to B^*$ between free monoids (resp.  each morphism
$\psi\colon A^* \to M$ from a free monoid to a finite monoid) admits a
uniquely defined continuous extension, $\hat\phi\colon \widehat{A^*}
\to \widehat{B^*}$ (resp.  $\hat\psi\colon \widehat{A^*} \to M$).

For example, consider the Cauchy sequence $(u^{n!})_n$,  where $u\in A^*$, which we discussed above.  This represents an element of $\widehat{A^*}$, which we will denote $u^{\omega}$\index{omega!$\omega$-power}.  Observe that for any morphism $\phi$ from $A^*$ into a finite monoid, the sequence $\hat\phi(u^{n!})$ is ultimately constant and equal to the unique idempotent power of $\phi(u)$, so in the notation we introduced earlier we have, very conveniently,
\begin{displaymath}
\hat\phi(u^{\omega})=(\phi(u))^{\omega}.
\end{displaymath}
We can similarly define $u^{\omega -1}$ as the element of $\widehat{A^*}$ represented by the Cauchy sequence $\hat\phi(u^{{n!}-1})$.

Finally, we note the strong connection between regular languages and 
free profinite monoids.

\begin{proposition}\label{SW:prop: regular and profinite}
    Let $A$ be an alphabet and let $L\subseteq A^*$.
    \begin{conditions}
    \item $L$ is regular if and only if its topological closure in $\widehat{A^*}$, $\overline L$, is clopen
    (i.e., open and closed), if and only if $L = K \cap A^*$ for some
    clopen set $K \subseteq \widehat{A^*}$.
    
    \item If $L$ is regular and $u\in\widehat{A^*}$, then the 
    following are equivalent:
    \begin{enumerate}
    \item[(i)]$u \in \overline L$;
    
    \item[(ii)] $\hat\phi(u) \in \phi(L)$ for every morphism $\phi$ from $A^*$
    to a finite monoid;
    
    \item[(iii)]$\hat\phi(u) \in \phi(L)$ for every morphism $\phi$ from $A^*$
    to a finite monoid recognizing $L$;
    
    \item[(iv)] $\hat\eta(u) \in \eta(L)$ where $\eta$ is the syntactic morphism
    of $L$.
    \end{enumerate}
    \end{conditions}
\end{proposition}

\subsection{Equations and lattices of languages}

We begin our study of families of regular languages with the simplest such family: a lattice  of
languages over a fixed alphabet.  In this chapter, we define a
\emph{lattice of languages}\index{lattice!of languages} over an alphabet $A$ to be a set of
languages over $A$ which is closed under finite union and finite
intersection, and which contains $A^*$ and $\emptyset$ (respectively,
the union and the intersection of an empty family of languages).

A \emph{profinite equation on $A$}\index{profinite!equation}\index{equation!profinite} is a pair $(u,v)$ of elements of
$\widehat{A^*}$, usually denoted by $u \to v$.  If $u,v\in A^*$, the
equation is called \emph{explicit}\index{equation!profinite!explicit}\index{explicit!profinite equation}.  A language $L \subseteq A^*$ is
said to \emph{satisfy} the equation $u \to v$, written $L \vdash u
\to v$, if
\begin{displaymath}
u \in \overline L \Longrightarrow v \in \overline L.
\end{displaymath}

\begin{remark}
    It is important to note that $u$,  $v$ and the words in $L$ are all defined over the
    same alphabet $A$.  In contrast to the identities we encountered in Section 1, in this definition, the letters occurring in
    $u$ and $v$ are \emph{not} considered as variables, to be replaced by
    arbitrary elements.  We will formally define
    identities in Section~\ref{SW:sec: identities and varieties}.
\end{remark}

The notion of equation is particularly relevant for regular
languages.  The following results directly from Proposition~\ref{SW:prop:
regular and profinite}.

\begin{proposition}\label{SW:prop: elementary equations}
    Let $L\subseteq A^*$ be regular and let $u,v\in \widehat{A^*}$.
    \begin{conditions}
    
    \item If $u,v \in A^*$, then $L \vdash u \to v$ if and only
    if $u\in L \Longrightarrow v\in L$.
    
    \item If $\eta$ is the syntactic morphism of $L$, then $L
    \vdash u\to v$ if and only if $\hat\eta(u) \in \eta(L)
    \Longrightarrow \hat\eta(v) \in \eta(L)$.
    \end{conditions}
\end{proposition}

Let $E$ be a set of equations on $A$.  We denote by $\mathcalL(E)$ the set
of regular languages in $A^*$ which satisfy all the equations in $E$.
It is immediately verified that this set is closed under unions and
intersections.  Further, both $\emptyset$ and $A^*$ satisfy every equation. So $\mathcalL(E)$ is a lattice.  The main theorem of
this section states that all lattices of regular languages arise this 
way.

\begin{theorem}\label{SW:thm lattice}
    Let $\mathcalL$ be a class of regular languages in $A^*$. Then 
    $\mathcalL$ is a lattice if and only if there exists a set $E$ of 
    profinite equations on $A$ such that $\mathcalL = \mathcalL(E)$.
\end{theorem}

We have already seen that one direction of this equivalence holds: 
every set of the form $\mathcalL(E)$ is a lattice. The proof of the 
converse is obtained after several steps. The first concerns the set 
of equations satisfied by a given language. If $L \subseteq A^*$, let
\begin{displaymath}
E_L = \left\{(u,v) \in \widehat{A^*} \times \widehat{A^*} \mid L \vdash 
u \to v\right\}.
\end{displaymath}

\begin{lemma}\label{SW:claimEL}
    If $L$ is regular, then $E_L$ is clopen.
\end{lemma}

\begin{proof}
By definition of the satisfaction of equations, we have
\begin{displaymath}
E_L = \left\{(u,v) \in \widehat{A^*} \times \widehat{A^*} \mid (u 
\not\in\overline{L}) \vee (v \in \overline{L})\right\} = 
\left(\overline{L}^c\times\widehat{A^*}\right) \cup 
\left(\widehat{A^*}\times\overline{L}\right).
\end{displaymath}
Lemma~\ref{SW:claimEL} follows from the fact that $\widehat{A^*}$, $\overline L$
and $\overline{L}^c$ are compact (since $L$ is regular).
\end{proof}

The proof of the next claim illustrates the crucial role played by the compactness of $\widehat{A^*}$.  Let $\mathcalL$ be a lattice of regular languages in $A^*$ and let $E_\mathcalL
= \bigcap_{L\in\mathcalL} E_L$.

\begin{lemma}\label{SW:claim compactness}
    Let $L$ be a regular language in $\mathcalL(E_\mathcalL)$: that is, $L$
    satisfies all the profinite equations satisfied by all the
    elements of $\mathcalL$.  Then there exists a finite subset $\mathcalK$ of
    $\mathcalL$ such that $L \in \mathcalL(E_\mathcalK)$.
\end{lemma}

\begin{proof}
By Lemma~\ref{SW:claimEL}, $E_L$ and each $E_K^c$ ($K\in \mathcalL$) are
open sets.  Moreover, if $(u,v)$ does not belong to any of the $E_K^c$
($K\in \mathcalL$), then $(u,v)$ belongs to each $E_K$, that is, every
language in $\mathcalL$ satisfies $u\to v$.  It follows that $L$ satisfies
$u\to v$ as well, that is, $(u,v) \in E_L$.  Therefore $E_L$ and the
$E_K^c$ ($K\in \mathcalL$) form an open cover of $\widehat{A^*}$.

By compactness, there exists a finite subcollection $\mathcalK$ of $\mathcalL$
such that $\widehat{A^*}$ is covered by $E_L$ and the $E_K^c$, $K\in
\mathcalK$.  It follows that $E_L$ contains the complement of
$\bigcup_{K\in\mathcalK}E_K^c$, namely the intersection
$\bigcap_{K\in\mathcalK}E_K$.  That is, $L$ satisfies all the equations
satisfied by the elements of $\mathcalK$, which establishes the claim.
\end{proof}

We are now ready to prove Theorem~\ref{SW:thm lattice}, by showing that
if $\mathcalL$ is a lattice of regular languages in $A^*$, then $\mathcalL =
\mathcalL(E_\mathcalL)$.  It is immediate by construction that $\mathcalL$ is
contained in $\mathcalL(E_\mathcalL)$.  Let us now consider a language $L \in
\mathcalL(E_\mathcalL)$.  By Lemma~\ref{SW:claim compactness}, we have $L \in
\mathcalL(E_\mathcalK)$ for a finite subset $\mathcalK$ of $\mathcalL$.

For each $u\in L$, let $\mathcalK(u)$ be the intersection of the 
languages $K\in \mathcalK$ containing $u$. Even though $L$ may be 
infinite, $\mathcalK(u)$ takes only finitely many values since $\mathcalK$ is 
finite. By definition of the $\mathcalK(u)$, we have $L \subseteq 
\bigcup_{u\in L}\mathcalK(u)$, a finite union.

Conversely, let $v\in \bigcup_{u\in L}\mathcalK(u)$.  Then there exists a
word $u\in L$ such that $v$ belongs to every $K\in \mathcalK$ containing
$u$.  That is, every $K\in \mathcalK$ satisfies the equation $u \to v$.
In other words, $u\to v$ lies in $E_\mathcalK$, and hence $L$ satisfies
that equation.  Since $u\in L$, it follows that $v\in L$.  Thus $L =
\bigcup_{u\in L}\mathcalK(u)$ and hence $L \in \mathcalL$, which concludes the
proof.


\subsection{More classes of languages: from lattices to varieties}\label{SW:sec: lattices 2 varieties}

Here we explore how classes of regular languages that are more
structured than lattices can be defined by more structured sets of
equations. We start with an elementary lemma.

\begin{lemma}
    Let $\mathcalL$ be a lattice of regular languages satisfying the 
    profinite equation $u\to v$.
    \begin{conditions}
    \item If $\mathcalL$ is closed under complementation, then $\mathcalL$ 
    also satisfies $v\to u$.
    
    \item If $\mathcalL$ is closed under quotients, then $\mathcalL$ satisfies the 
    equations $xuy \to xvy$, for all $x,y\in \widehat{A^*}$.
\end{conditions}
\end{lemma}

\begin{proof}
It follows from the definition of equations that $L$ satisfies $u\to 
v$ if and only if its complement satisfies $v\to u$. The first part 
of the claim follows immediately.

It is also elementary that, if $x,y\in A^*$ and $x\inv Ly\inv \vdash
u\to v$, then $L\vdash xuy\to xvy$.  Thus, if $\mathcalL$ is closed under
quotients, then $\mathcalL$ satisfies all the equations $xuy \to xvy$ with
$x,y\in A^*$.  This holds also if $x,y\in \widehat{A^*}$ since $E_\mathcalL$
is closed and $A^*$ is dense in $\widehat{A^*}$.
\end{proof}

We now extend the notion of profinite equations as follows: if $u,v\in
\widehat{A^*}$, we say that a language $L$ satisfies the
\emph{symmetrical} equation\index{equation!symmetrical} $u \leftrightarrow v$ if $L$ satisfies
both $u\to v$ and $v\to u$.

We also say that a language $L$ satisfies the \emph{profinite
inequality}\index{profinite!inequality} $v \le u$ if it satisfies all the equations of the form
$xuy \to xvy$ with $x,y \in \widehat{A^*}$, and it satisfies the
\emph{profinite equality}\index{profinite!equality} $u = v$ if if satisfies both $u\le v$ and
$v\le u$.
The verification of the following corollary is now elementary.

\begin{corollary}\label{SW:cor: boolean and quotients}
    Let $\mathcalL$ be a set of regular languages in $A^*$.
    \begin{conditions}
    
    \item Then $\mathcalL$ is a boolean algebra if and only if $\mathcalL
    = \mathcalL(E)$ for some set $E$ of symmetrical profinite equations on $A$.
    
    \item $\mathcalL$ is a lattice closed under quotients if and
    only if $\mathcalL = \mathcalL(E)$ for some set $E$ of profinite
    inequalities on $A$.
    
    \item $\mathcalL$ is a boolean algebra closed under quotients if
    and only if $\mathcalL = \mathcalL(E)$ for some set $E$ of profinite
    equalities on $A$.
    \end{conditions}
\end{corollary}

\subsection{Identities and varieties}\label{SW:sec: identities and varieties}

We now come to the historically and mathematically important class of
varieties.  Varieties of languages were defined in Section~\ref{SW:sec:
varieties} but we will not  use this definition here.  In fact, in
the course of this section, we will give an alternate, equivalent
definition of varieties.

An important difference between varieties and the lattices of languages over a fixed
alphabet discussed so far in Section~\ref{SW:sec: eqs}, is that a
variety $\mathcalV$ consists of a collection of lattices $A^*\mathcalV$, one
for each finite alphabet $A$.  More generally, we define a
\emph{class of regular languages}\index{class of languages} $\mathcalV$ to be an operator which
assigns to each finite alphabet $A$, a family $A^*\mathcalV$ of regular
languages in $A^*$.

First, we prove a technical lemma.

\begin{lemma}\label{SW:lemma: inverse morphism}
    Let $\phi\colon A^* \to B^*$ be a morphism, $L\subseteq B^*$ and
    $u,v \in \widehat{A^*}$.
    \begin{conditions}
    \item $\hat\phi(u) \in \overline{L}$ if and only if $u\in
    \overline{\phi\inv(L)}$.
    
    \item $L$ satisfies $\hat\phi(u) \to
    \hat\phi(v)$ if and only if $\phi\inv(L)$ satisfies $u\to v$.
\end{conditions}\end{lemma}

\begin{proof}
The first statement is trivial if $u,v\in A^*$: indeed, $\phi$ and
$\hat\phi$ coincide on words, and the intersection of $\overline{L}$
(resp.  $\overline{\phi\inv(L)}$) with $A^*$ (resp. $B^*$) is $L$ (resp.
$\phi\inv(L)$). The extension to the case where $u,v \in 
\widehat{A^*}$ is obtained by density.


The second statement follows immediately from the first and the 
definition of profinite equations.
\end{proof}

We extend the notion of profinite equations, this time to 
profinite identities, to permit the treatment of classes of
regular languages instead of lattices of regular languages over a
fixed alphabet.  Since there is no alphabet of reference anymore, we
will usually denote by $X$ the alphabet over which profinite
identities are written.

Let $\mathcalC$ be a composition-closed class of morphisms between free monoids, $u,v\in
\widehat{X^*}$ and $L\subseteq A^*$, where $X$ and $A$ are finite, but
possibly different alphabets.  We say that $L$
\emph{$\mathcalC$-identically satisfies} $u\to v$ if, for each morphism
$\phi\colon X^*\to A^*$ in $\mathcalC$, $L$ satisfies $\hat\phi(u) \to
\hat\phi(v)$.  We say that a class of regular languages $\mathcalV$
$\mathcalC$-identically satisfies an equation if $A^*\mathcalV$ does, for each
finite alphabet $A$.

The following statement is a direct application of Lemma~\ref{SW:lemma:
inverse morphism}.

\begin{corollary}\label{SW:cor: inverse morphisms}
    Let $\mathcalV$ be a class of regular languages, let $\mathcalC$ be a
    family of morphisms between free monoids closed under composition, such that
    whenever $\phi\colon A^*\to B^*$ is in $\mathcalC$ and $L\in
    B^*\mathcalV$, then $\phi\inv(L) \in A^*\mathcalV$.
    
    If $X^*\mathcalV$ satisfies the profinite equation $u\to v$ (with
    $u,v\in\widehat{X^*}$), then $\mathcalV$ $\mathcalC$-identically satisfies
    $u\to v$.
\end{corollary}

Using the notions introduced in Section~\ref{SW:sec: lattices 2
varieties}, we say that $L$ satisfies the \emph{profinite}
\emph{$\mathcalC$-identity}\index{profinite!$\mathcalC$-identity} $u = v$ (resp.  \emph{profinite ordered
$\mathcalC$-identity}\index{profinite!$\mathcalC$-identity!ordered} $u\le v$) if $L$ $\mathcalC$-identically satisfies $u =
v$ (resp.  $u\le v$).  If $E$ is a set of profinite equations and for
each finite alphabet $A$, $A^*\mathcalV$ is the set of regular languages
in $A^*$ which $\mathcalC$-identically satisfy the elements of $E$, we say
that the resulting class of regular languages $\mathcalV$ is
\emph{$\mathcalC$-defined by} $E$.

Let us now define (positive) $\mathcalC$-varieties: a class $\mathcalV$ of
regular languages is a \emph{positive $\mathcalC$-variety}\index{variety!$\mathcalC$-}\index{variety!$\mathcalC$-!positive} (resp.  a
\emph{$\mathcalC$-variety}) \emph{of languages} if each $A^*\mathcalV$ is
a lattice (resp.  a boolean algebra) closed under quotients and if,
for each $\phi\colon A^*\to B^*$ in $\mathcalC$ and each $L\in B^*\mathcalV$,
we have $\phi\inv(L)\in A^*\mathcalV$.

If $\mathcalC$ is the class of all morphisms between free monoids, we drop
the prefix $\mathcalC$ and simply talk of (ordered) profinite identities and (positive) varieties of languages.

Collecting Corollaries~\ref{SW:cor: boolean and quotients} and~\ref{SW:cor:
inverse morphisms}, we have the following characterizations.

\begin{theorem}\label{SW:thm: identities and varieties}
    Let $\mathcalV$ be a class of regular languages and let $\mathcalC$ be a
    composition-closed class of morphisms between free monoids.  Then
    $\mathcalV$ is a positive $\mathcalC$-variety (resp.  a $\mathcalC$-variety)
    if and only if $\mathcalV$ is $\mathcalC$-defined by a set of profinite
    ordered $\mathcalC$-identities (resp.  profinite $\mathcalC$-identities).
\end{theorem}

\begin{remark}
    In Section~\ref{SW:sec: varieties}, we gave a different definition of
    varieties of languages, and Theorem~\ref{SW:thm:eilenberg} stated
    that it was equivalent to the definition given above.  We will
    prove this equivalence in Section~\ref{SW:sec:eilenbergreiterman}
    below, thus formally reconciling the two definitions.
\end{remark}

\subsection{Eilenberg's and Reiterman's theorems}\label{SW:sec:eilenbergreiterman}

We note that (in)equalities  can be
interpreted in the (ordered) syntactic monoid of a language.  Let $L$
be a regular language in $A^*$ and let $u,v\in \widehat{A^*}$.  By
Proposition~\ref{SW:prop: elementary equations}, if $\eta$ is the
syntactic morphism of $L$, then $L \vdash v\le u$ if and only if
$\hat\eta(v) \le_L \hat\eta(u)$.

Thus membership of a regular language $L$ in a lattice of regular languages closed under quotients is characterized by properties of the syntactic morphism of $L$.

We can also interpret identities in abstract finite ordered monoids---that is, finite monoids in which there is a partial order $\leq$ compatible with multiplication:  If $u,v\in\widehat{X^*}$, we say that a finite ordered monoid $M$ satisfies the profinite identity $u\leq v$ if for every morphism $\phi\colon X^*\to M$ we have $\hat\phi(u)=\hat\phi(v)$. Likewise a monoid $M$ satisfies the profinite identity $u=v$ if for each such $\phi$ we have $\hat\phi(u)=\hat\phi(v)$. We extend this notion further to ${\cal C}$-satisfaction of identities.  We call a morphism $\phi\colon A^*\to M$, where $M$ is finite and $\phi$ maps onto $M$, a \emph{stamp}\index{stamp}.  We also define \emph{ordered stamps}\index{stamp!ordered} as morphisms from a free monoid $A^*$ onto an ordered finite monoid.  (Such morphisms are automatically order-preserving if we consider the trivial ordering on $A^*$ in which $w_1\leq w_2$ if and only if $w_1=w_2$.)  Let $\mathcalC$ be a class of morphisms between finitely generated free monoids that is closed under composition and that contains all the length-preserving morphisms.  We say that the ordered stamp $\phi\colon A^*\to (M,\le)$ \emph{$\mathcalC$-satisfies the profinite identity} $u\leq v$ with $u,v\in\widehat{X^*}$ if and only if for all morphisms $\psi\colon X^*\to A^*$ with $\psi\in\mathcalC$, we have $\hat\phi\hat\psi(u)\leq\hat\phi\hat\psi(v)$.  We similarly define ${\cal C}$-satisfaction of identities $u=v$ by (not necessarily ordered) stamps.

We have already defined pseudovarieties of finite monoids in Section 1.  We can extend this definition to define $\mathcalC$-pseudovarieties of stamps\index{pseudovariety!$\mathcalC$-pseudovariety of stamps}.  We call a collection $\psvV$ of stamps a \emph{$\mathcalC$-pseudovariety} if it satisfies the following two conditions:
\begin{enumerate}
\item[(i)] If $\phi\colon A^*\to M$ is in $\psvV$, $\psi\colon B^*\to A^*$ is in $\mathcalC$, and $\eta$ is a morphism from $Im(\phi\psi)$ onto a finite monoid $N$, then
$\eta\phi\psi\colon B^*\to N$ is in $\psvV$.
\item[(ii)] If $\phi_i\colon A^*\to M_i$ are in $\psvV$ for $i=1,2$, then
$\phi_1\times\phi_2\colon A^*\to Im(\phi_1\times\phi_2)\subseteq M_1\times M_2$ is in $\psvV$.
\end{enumerate}
If we restrict the morphisms occurring in these definitions to order-preserving morphisms or ordered monoids, we obtain the definition of \emph{ordered $\mathcalC$-pseudovarieties of stamps}.  Ordinary pseudovarieties coincide with $\mathcalC$-pseudovarieties in the case where $\mathcalC$ contains all morphisms between finitely-generated free monoids.


We say that a class $\psvV$ of finite (ordered) monoids is \emph{defined} by a set $E$ of identities (written $\psvV = \llbrack E\rrbrack$) if $\psvV$ consists of all the finite (ordered)monoids that satisfy all of the identities in $E$.  Similarly, we say that a family $\psvV$ of stamps is $\mathcalC$-defined by $E$ (we write $\psvV = \llbrack E\rrbrack_{\calC}$) if $\psvV$ consists of all the stamps that $\mathcalC$-satisfy these identities.

Further if $\psvV$ is a class of monoids or stamps, ordered or unordered, we define the corresponding class $\mathcalV$ of languages by setting $L\in A^*\mathcalV$ if and only if $\synt(L)\in\psvV$ (if $\psvV$ is a class of monoids) or $\eta_L\in\psvV$ (if $\psvV$ is a class of stamps).  We write $\psvV\mapsto\mathcalV$ to denote this correspondence. 

This leads us to a restatement of Eilenberg's Theorem\index{Eilenberg's
Theorem}, Theorem~\ref{SW:thm:eilenberg} above, as well as its generalization to
${\cal C}$-varieties, and allows us to prove it simultaneously with Reiterman's
Theorem\index{Reiterman's Theorem}.

\begin{theorem}
The following statements hold.
\begin{conditions}
\item \emph{(Eilenberg's Theorem)} If $\psvV$ is a pseudovariety (respectively $\mathcalC$-pseudovariety, ordered pseudovariety) and $\psvV \mapsto\mathcalV$,
then $\mathcalV$ is a variety of languages (respectively $\mathcalC$-variety of languages, positive variety of languages) and in each case this gives a one-to-one correspondence between pseudovarieties and varieties of languages.
\item \emph{(Reiterman's Theorem)} A class $\psvV$ of monoids (stamps, ordered monoids) is a pseudovariety (respectively $\mathcalC$-pseudovariety, ordered pseudovariety) if and only if it is defined ($\mathcalC$-defined) by a set of profinite identities. 
\end{conditions}
\end{theorem}

In the argument we sketch below, we confine our\-selves to the case of ord\-inary monoids, but everything generalizes in an entirely straightforward fashion to ordered monoids and stamps. 
The key to the proofs of both parts of the theorem is Theorem~\ref{SW:thm: identities
and varieties} above, along with the following elementary
but very useful lemma, already brought to the reader's attention in
Section~\ref{SW:subsection:j_1}.

\begin{lemma}\label{SW:lemma: syntactic generation}
    Let $\phi\colon A^*\to M$ be a morphism into a finite monoid. 
    Then $M$ divides the direct product of the syntactic monoids
    of the languages $\phi\inv(m)$, $m\in M$.
\end{lemma}

\begin{proof}
For each $m\in M$, let $\eta_m\colon A^* \to \synt(\phi\inv(m))$ be 
the syntactic morphism of $\phi\inv(m)$. It suffices to show that for 
each $u,v\in A^*$, $\eta_m(u) = \eta_m(v)$ for each $m\in M$ implies 
$\phi(u) = \phi(v)$.

Indeed, let $m = \phi(u)$. Then $u\in \phi\inv(m)$ and since 
$\eta_m(v) = \eta_m(u)$, we have $v\in \phi\inv(m)$, $\phi(v) = m = 
\phi(u)$.
\end{proof}

\begin{corollary}\label{SW:cor: syntactic generation}
    Every pseudovariety of  monoids is generated by the 
    syntactic  monoids it contains.
\end{corollary}

\begin{proof}
The result follows directly from Lemma~\ref{SW:lemma: syntactic
generation}, since $M$ recognizes each $\phi\inv(M)$ ($m\in M$): thus each
$\synt(\phi\inv(m))$ divides $M$ and hence lies in the pseudovarieties 
containing $M$.
\end{proof}

Now let $\mathcalV$ be a variety of languages and let $E$ be a set of 
profinite identities defining $\mathcalV$. Let also $\psvV$ be  the 
class of finite monoids satisfying the profinite identities in $E$. It 
is easily verified that $\psvV$ is a pseudovariety.

Moreover, if $L$ is a regular language in $A^*$, we have $L\in 
A^*\mathcalV$ if and only if $L \vdash E$, if and only if $\synt(L)$ 
satisfies the profinite identities in $E$, if and only if $\synt(L)\in \psvV$.

Thus $\psvV \mapsto \mathcalV$ in the correspondence described in
Section~\ref{SW:sec: varieties}. If $\psvW$ is another pseudovariety such 
that $\psvW\mapsto \mathcalV$, then $\psvV$ and $\psvW$ contain the same syntactic 
monoids, and Corollary~\ref{SW:cor: syntactic generation} shows that $\psvV 
= \psvW$. This establishes Eilenberg's Theorem.

For Reiterman's Theorem, we start with a pseudovariety $\psvV$ and 
consider the associated variety of languages $\mathcalV$. The above 
reasoning shows that $\psvV$ is defined by any set of profinite identities 
which, seen in the setting of classes of languages, defines $\mathcalV$.


Note that these proofs are different from the classical proofs of
Eilenberg's theorem, in \cite{Eilenberg:1976} or \cite{Pin:1986}, and of
Reiterman's theorem, in \cite{Almeida:1994}, \cite{Pin&Weil:1996} or
\cite{Reiterman:1982}.

\subsection{Examples of varieties}



We now look at some concrete instances of varieties, revisiting our examples from Section 1, among others, in light of the theory presented above.  In doing so, we will work from both sides of the correspondence between pseudovarieties and varieties of languages, at times beginning with a variety of languages, at others with a property of a class of finite monoids.

\subsubsection{Idempotent and commutative monoids}\label{SW:sec: j1 again}

We begin, as before, with the variety of languages corresponding to the pseudovariety $\psvJ_1$. For each finite alphabet $A$, let $A^*{\mathcal J}_1$ be the smallest boolean-closed family of subsets of $A^*$  that
contains all the languages $B^*$, where $B\subseteq A$. Equivalently, it is the smallest boolean-closed set containing all the $A^*aA^*$ ($a\in A$). Putting it again differently, $A^*{\mathcal J}_1$ is precisely the family of languages $L$ in
$A^*$ for which membership of a word $w$ in $L$ depends only  on the set $\alpha(w)$ of letters of $w$.  This is
because
\begin{displaymath}
\{v\in A^* \mid \alpha(v)=\alpha(w)\} = \alpha(w)^*\backslash \bigcup_{B\subsetneq \alpha(w)}B^*.
\end{displaymath}
Observe that for all $a\in A$ and $B\subseteq A$,
\begin{displaymath}
a^{-1}B^* = B^*a^{-1} = \left\{\begin{array}{ll}\emptyset & \mbox{if $a\notin B$}\\
                                                                        B^* & \mbox{if $a\in B$.}
                                                                        \end{array}\right .
\end{displaymath}
Further, if $C$ is another finite alphabet and $\phi\colon C^*\to A^*$ is a morphism,
\begin{displaymath}
\phi^{-1}(B^*) = (C\cap \phi^{-1}(B))^*.
\end{displaymath}
Left and right quotient and inverse image under morphisms all commute with
boolean operations. So these two observations imply, independently of any
algebraic considerations, that ${\mathcal J}_1$ is a variety of languages, and
thus, by Theorem~\ref{SW:thm: identities and varieties}
 is defined by a set of profinite identities. Further, from our proof of Eilenberg's Theorem, the same set of identities defines the corresponding pseudovariety of finite monoids.
 
Of course, we have already exhibited these identities, but let us see what they look like in the context of our equational theory.  Let $X=\{x,y\}$, and let $A$ be any finite alphabet. Every language $L\in A^*{\mathcal J}_1$ satisfies  the identities  $xy=yx$ and $x^2=x$, since for any morphism $\phi\colon X^*\to A^*$ and any $u,v\in A^*, $  $\alpha (u\phi(xy)v)=\alpha(u\phi(yx)v)$, and $\alpha(u\phi(x^2)v)=\alpha(u\phi(x)v)$.  Conversely, suppose $L\subseteq A^*$ satisfies these identities. We will show $L\in A^*{\mathcal J}_1$:  Let $w, w'\in B^*$, with $w\in L$ and  $\alpha(w)=\alpha(w')$.  We claim $w'\in L$.  Since $\alpha(w)=\alpha(w')$, we can transform both $w$ and $w'$ into a common normal form $w''$ by successively interchanging adjacent letters until the word is sorted (with respect to some total ordering on $A$) and then  replacing occurrences of $aa$ by $a$, where $a\in A$.  Interchanging adjacent letters entails replacing $ua_1a_2v$ by $ua_2a_1v$, where $u,v\in A^*$ and $a_1,a_2\in A$.  Since $L$ satisfies the identity $xy=yx$, if $ua_1a_2v\in L$ then $ua_2a_1v\in L$ (using the morphism $\phi\colon X^*\to A^*$ that maps $x, y$  to $a_1,a_2$, respectively.).  Similarly, replacing $aa$ by $a$ preserves membership in $L$, since $L$ satisfies the identity $x^2 = x$.  Thus ${\mathcal J}_1$ is defined by this pair of identities.  It follows that the corresponding pseudovariety $\psvJ_1$ of finite monoids is defined by the same pair of identities, and thus consists of the idempotent and commutative monoids.

\subsubsection{Piecewise-testable languages}

Now let us consider the piecewise-testable langua\-ges\index{language!piecewise-testable} of Section~\ref{SW:subsection:j}. We denote the family of piecewise-testable languages over a finite alphabet $A$ by $A^*{\mathcal J}$.  Let us look at the profinite identities satisfied by these languages.   As observed earlier (Section~\ref{SW:sec: profinite monoid}), if $u\in X^*$ then the sequence $(u^{n!})_n$ is a Cauchy sequence whose limit is written $u^{\omega}$. Moreover, for any morphism $\phi\colon X^*\to A^*$, where $A$ is a finite alphabet, $\hat\phi(u^{\omega})=(\hat\phi(u))^{\omega}$ (the idempotent power of $\hat\phi(u)$).  Now let $X=\{x, y\}$.  We claim that every piecewise-testable language $L$ over $A^*$ satisfies the profinite identities
\begin{displaymath}
(xy)^{\omega}x=(xy)^{\omega}=y(xy)^{\omega}.
\end{displaymath}
This is equivalent to saying that for all $s,t,u,v\in A^*$,
\begin{displaymath}
s(tu)^{\omega}tv\in\overline L\Leftrightarrow s(tu)^{\omega}v\in\overline L\Leftrightarrow su(tu)^{\omega}v\in\overline L.\end{displaymath}
Now fix an integer $k>0$.  For sufficiently large values of $n$, the words 
\begin{displaymath}
s(tu)^{n!}tv, s(tu)^{n!}v, su(tu)^{n!}v
\end{displaymath}
contain the same subwords of length $k$.  Since $L$ is piecewise-testable,  for sufficiently large $n$, all but finitely many of the terms of the three sequences are either all in $L$ or all outside of $L$. Since $\overline L$ is clopen, the three respective  limits are either all in $\overline L$ or all outside $\overline L$.

Thus, as we showed in Section~\ref{SW:sec:eilenbergreiterman}, the syntactic monoid of any piecewise testable language satisfies these same profinite identities. We arrive again at the observation that the syntactic monoid of every piecewise-testable language satisfies the identities $(xy)^{\omega}x=(xy)^{\omega}=y(xy)^{\omega}$. That these identities define the pseudovariety $\psvJ$ of finite ${\mathcal J}$-trivial monoids is simple to establish.  That they completely characterize the variety of piecewise-testable languages is the deep content of Simon's Theorem~\cite{Simon:1975}.  

\subsubsection{Group languages}

Similarly, the pseudovariety $\psvG$ of finite groups is defined by the profinite identity $x^{\omega} = 1$.  As a consequence, the corresponding variety ${\mathcal G}$ of languages is defined by the same profinite identity.  In contrast to the other  examples presented here, we do not possess a simple description of ${\mathcal G}$ in terms of basic operations on words.

\subsubsection{Left-zero semigroups}\index{semigroup!left-zero}

We already appealed to Eilenberg's Theorem in Section 1 to show that the class ${\mathcal K}_1$ is not a variety of languages. But we can show here that it is a ${\mathcal C}$-variety for a slightly restricted class ${\mathcal C}$ of morphisms. Let ${\mathcal C}_{ne}$ denote the class  of \emph{non-erasing} morphisms\index{morphism!non-erasing} between finitely-generated free monoids--those $\phi\colon A^*\to B^*$ such that for all $a\in A$, $\phi(a)\neq 1$. Let $L\in A^*{\mathcal K}_1$.  If $s,t,u,v\in A^*$, and $t,u\neq 1$, then $stuv\in L$ if and only if $stv\in L$.   Moreover, this property of $L$ characterizes membership in $A^*{\mathcal K}_1$.  One way to state this property is that  the variety of languages ${\mathcal K}_1$ is defined by the ${\mathcal C}_{ne}$-identity $xy=x$.  Equivalently, the corresponding ${\mathcal C}_{ne}$-pseudovariety  $\psvK_1$ of stamps is defined by the same ${\mathcal C}_{ne}$-identity.  This means  $(\phi\colon A^*\to M)\in\psvK_1$ if $\phi(uv)=\phi(u)$, for $u,v\in A^+$.  

Alternatively, one may  consider, instead of the ${\mathcal C}_{ne}$-pseudovariety generated by the syntactic morphisms of languages in ${\mathcal K}_1$, the pseudovariety of finite \emph{semigroups} generated by the images of nonempty words under the syntactic morphisms. This was the approach originally taken, but here we prefer to emphasize that all these many different flavors of pseudovarieties can be treated in the same general setting.

\subsubsection{Quasiaperiodic stamps}\index{stamp!quasi-aperiodic}

Whenever we have a morphism $\phi\colon A^*\to M$, the family of sets
\begin{displaymath}
\{\phi(A^s) \mid s>0\}
\end{displaymath}
forms a subsemigroup of the power set semigroup ${\mathcal P}(M)$.  As this is a finite cyclic semigroup, generated by $\phi(A)$, it contains a unique idempotent. Thus there is some $s>0$ such that $\phi(A^s)=\phi(A^{2s})$, so that $\phi(A^s)$ is a subsemigroup of $M$.  We call this the \emph{stable semigroup}\index{semigroup!stable} of $\phi$.  Let ${\bf QA}$ denote the set of morphisms $\phi$ from a free finitely-generated monoid onto a finite monoid such that $\phi$ is surjective, and the stable semigroup of $\phi$ is aperiodic.

We claim ${\bf QA}$ is a ${\mathcal C}_{lm}$-pseudovariety of stamps, where ${\mathcal C}_{lm}$ consists of morphisms $\psi\colon A^*\to B^*$ between finitely generated free monoids such that all $\psi(a)$, where $a\in A$, are nonempty words having the same length.  (The letters \emph{lm} stand for \emph{length-multiplying}, since the lengths of all words in $A^*$ are multiplied by a constant factor when $\psi$ is applied.)  To see this, suppose $(\phi\colon B^*\to M)\in{\bf QA}$, and $(\psi\colon A^*\to B^*)\in {\mathcal C}_{lm}$.  Let $\phi(B^s)$ be the stable semigroup of $\phi$, $\phi\psi(A^t)$ the stable semigroup of $\phi\psi\colon A^*\to Im(\phi\psi)$, and $k$ the length of each $\psi(a)$ for $a\in A$. Then $\phi\psi(A^t)=\phi\psi(A^{st})\subseteq\phi(A^{kst})=\phi(A^s)$, and thus the stable semigroup of $\phi\psi$ is also aperiodic.  Further, if the stable semigroups $\phi_j(A^{s_j})$ of stamps $\phi_j\colon A^*\to M_j$, for $j=1,2$, are aperiodic, then the stable semigroup of $\phi_1\times\phi_2$ is contained in $\phi_1(A^{s_1})\times\phi_2(A^{s_2})$, and is therefore aperiodic.Thus  ${\bf QA}$ is a ${\mathcal C}_{lm}$-pseudovariety, and is accordingly defined by a set of profinite ${\mathcal C}_{lm}$-identities.  What does it mean for a stamp $\phi\colon A^*\to M$ to satisfy a ${\mathcal C}_{lm}$ identity $u=v$?  In such an identity, $u$ and $v$ are elements of $\widehat{X^*}$ for some finite alphabet $X$.  The identity is satisfied if for every morphism $\psi\colon X^*\to A^*$ in ${\mathcal C}_{lm}$, $\hat \phi \hat\psi(u)=\hat\phi\hat\psi(v)$.  Informally, this says that so long as we replace the letters in $u$ and $v$ by elements of $A^+$ that all have the same length, the images in $M$ are identical.  We claim that {\bf QA} is defined by the single profinite ${\mathcal C}_{lm}$-identity
\begin{displaymath}
(x^{\omega-1}y)^{\omega} = (x^{\omega-1}y)^{\omega+1}.
\end{displaymath}

Let us prove this.  First, we show that ${\bf QA}$ satisfies the identity.  Let $(\phi\colon A^*\to M)\in{\bf QA}$, and choose $p>0$ such that for all $m\in M$, $m^p$ is idempotent. We then also have $m^{ps}$ idempotent for all $m\in M$, where $\phi(A^s)$ is the stable semigroup of $\phi$. If the identity is \emph{not} satisfied, then there exist words $u$ and $v$ in $B^*$, both of length $k>0$, such that
\begin{displaymath}
(\phi(u^{ps-1}v))^{ps}\neq (\phi(u^{ps-1}v))^{ps+1}.
\end{displaymath}
Thus $\{(\phi(u^{ps-1}y))^{ps+r} \mid r\geq 0\}$ is a nontrivial group in $\phi((A^s)^+)=\phi(A^s)$, contradicting membership in ${\bf QA}$.  Conversely, suppose a stamp 
$\phi\colon A^*\to M$ satisfies the identity.  Suppose the stable semigroup $\phi(A^s)$ contains a group element $g=\phi(u)$, with $|u|=s$.  Let $e=\phi(v)$, where $|v|=s$ is the identity of this group. Since $\phi$ satisfies the identity,
\begin{displaymath}
e=\phi((u^{\omega-1}v)^{\omega})=\phi((u^{\omega-1}v)^{\omega+1})=g^{-1},
\end{displaymath}
so every group in $\phi(A^s)$ is trivial.

We introduced the ${\mathcal C}_{lm}$-pseudovariety ${\bf QA}$ in Section 1 in quite different terms, by giving a logical description of the corresponding ${\mathcal C}_{lm}$-variety of languages.  We will show in Section~\ref{SW:sec: logic} that they do in fact correspond.

\subsubsection{$\Sigma_1$-languages}\index{language!$\Sigma_1$}

As in Section~\ref{SW:sssec: sigma1}, we denote by $A^*{\mathcal J}^+$ the family of languges over $A$ defined by $\Sigma_1$ sentences.  Languages in this family are precisely the finite unions of the languages $L_v$, where $v\in A^*$.   We claim that ${\mathcal J}^+$ is defined by the profinite ordered identity $x\leq 1$.  A language $L$ satisfies this identity if and only if for all $u,v,w\in A^*$, whenever $uw\in L$, then $uvw\in L$.  Clearly, each $L_v$ satisfies this identity.  We must show, conversely, that any language satisfying this identity is a finite union of $L_v$ for various $v\in A^*$.  Certainly, if $L$ satisfies the identity and $v\in L$, then $L_v\subseteq L$, so that 
\begin{displaymath}
L=\bigcup_{v\in L} L_v.
\end{displaymath}
We need to show that this can be replaced by a finite union.  Let $T$ consist of the \emph{subword-minimal} elements of $L$, that is, those $v\in L$ such that no proper subword of $v$ is in $L$.  Then
\begin{displaymath}
L=\bigcup_{v\in T}L_v.
\end{displaymath}
We now invoke a theorem of G.Higman~\cite{Higman:1952}: The subword ordering in $A^*$ has no infinite antichains: That is, any set $T$ of words in which no element is a strict subword of another element is finite.  

The corresponding ordered pseudovariety $\psvJ^+$ consequently consists of all partially ordered finite monoids for which the identity $1$ is the maximum element, and thus a language belongs to $A^*{\mathcal J}^+$ if and only if its ordered syntactic monoid satisfies this property.

\subsubsection{Languages with zero}  

All of our examples so far have concerned some flavor of varieties of languages, language families that are defined across all finite alphabets and are closed under inverse images of morphisms between free monoids.  Part of the great novelty of the equational theory of Gehrke \emph{et al.} \cite{Gehrke&Grigorieff&Pin:2008} presented here is that it applies to language classes with weaker closure properties.  Here we give a simple example.

We say a regular language $L\subseteq A^*$ is a \emph{language with zero}\index{language!with zero} if $\synt(L)$ has a zero.  This is equivalent to saying that there is a two-sided ideal $J$ in $A^*$ such that either $J\subseteq L$ or $L\cap J=\emptyset$.  This property is easily seen to be closed under boolean operations and quotients.  It is, not, however,  closed under inverse images of any composition-closed class ${\mathcal C}$ of morphisms that contains the length-preserving morphisms.  Indeed, let $L\subseteq A^*$ be any regular language without a zero, and let $b$ be a new letter.  Then, viewed as a subset of $(A\cup\{b\})^*$, $L$ has a zero, so this class is not closed under the inverse image of the length-preserving morphism that embeds $A^*$ in $(A\cup\{b\})^*$.  Nonetheless, by our Corollary~\ref{SW:cor: boolean and quotients}, this class of languages is defined by a set of profinite inequalities.

We now exhibit such a set of inequalities.  We start by defining three sequences of words in $A^*$.  Let
\begin{displaymath}
u_1, u_2,\ldots
\end{displaymath}
be any enumeration of the elements of $A^*$,
let 
\begin{displaymath}
v_n=u_1\cdots u_n,
\end{displaymath}
and 
\begin{displaymath}
w_1=1, w_{n+1}=(w_nv_nw_n)^{n!}.
\end{displaymath}
Look at the image of the $w_i$ under  a surjective morphism $\phi\colon A^*\to M$, where $M$ is finite.  Since every $u\in A^*$ occurs as a factor of all but finitely many $w_i$, almost all $\phi(w_i)$ are in the minimal ideal  $K$ of $M$.  Since for all $m\in M$, $m^{n!}$ is idempotent for sufficiently large $n$, almost all $\phi(w_i)$ are idempotents in the minimal ideal of $M$.  Finally, if $\phi(w_i)$ is such an idempotent $e$, then $\phi(w_{i+1})$ is an idempotent in $eKe$, and so is itself equal to $e$.  Thus for every finite monoid, the sequence $(\phi(w_n))_n$ is convergent, so $(w_n)_n$
 converges to an element $\rho_A$ of $\widehat{A^*}$, such that $\hat\phi(\rho_A)$ is an idempotent in the minimal ideal of $\phi(A^*)$.

Suppose $L\subseteq A^*$ has a zero.  Then the minimal ideal of $\synt(L)$ consists of this 0 alone, so if $\eta$ is the syntactic morphism of $L$ and $a\in A$, $\hat\eta(\rho_A)=\hat\eta(a\rho_A)=\hat\eta(\rho_Aa)$.  Thus $L$ satisfies the equalities
\begin{displaymath}
a\rho_A=\rho_A=\rho_Aa
\end{displaymath}
for all $a\in A$.  Conversely, if $L$ satisfies these equalities, then the minimal ideal of $\eta(A^*)$ contains just one element, so $L$ is a language with zero.  So these equalities define the class of languages with zero.

\subsubsection{Languages defined by density}\label{SW:sec: density}

Say that a language $L \subseteq A^*$ is \emph{dense}\index{language!dense} if every word of $A^*$ occurs as a factor of a word in $L$, that is, $L \cap A^*uA^* \ne \emptyset$ for every $u\in A^*$. The set consisting of $A^*$ and the non-dense languages forms a quotient-closed lattice, which is defined by the profinite inequalities $x\le 0$ ($x\in A^*$)---this is short for $a\rho_A = \rho_Aa = \rho_A$ for every $a\in A$ and $x\le \rho_A$ for every $x\in A^*$ \cite{Gehrke&Grigorieff&Pin:2008}.

Now define the \emph{density of a language} $L$ as the function $d_L(n)$ which counts the number of words of length $n$ in $L$. A language with bounded density (also called \emph{slender}\index{language!slender}) is easily seen to be a finite union of languages of the form $xu^*y$ ($x,u,y\in A^*$). Similarly, a language of polynomial density, also called \emph{sparse}\index{language!sparse}, can be shown to be a finite union of languages of the form $u_0^*v_1u_1^*\cdots v_nu_n^*$ where the $u_i$ and $v_j$ are in $A^*$. Together with $A^*$, the set of slender (resp. sparse) languages in $A^*$ forms a quotient-closed lattice of languages, for which defining profinite inequalities can be found in \cite{Gehrke&Grigorieff&Pin:2008}.

\subsection{Deciding membership in an equationally defined class of 
languages}

We are often interested in decision problems for families of regular languages:  We say that a family ${\mathcal F}$ of regular languages over a finite alphabet $A$ is \emph{decidable} if there is an algorithm that, given a regular language in $L\subseteq A^*$ as input, determines whether $L\in{\mathcal F}$.
Here a regular language $L$  is `given' by specifying a DFA that recognizes $L$, or some other formalism (\emph{e.g.,} regular expression, logical formula) from which a DFA can be effectively computed.  The problem arises, for example, if we are looking for a test of whether a given language is expressible in some logic for defining regular languages.  (See Section~\ref{SW:sec: logic}.)

We can similarly define decidable families of finite monoids: Such a family ${\mathcal F}$ is decidable if there is an algorithm that, given the multiplication table for a finite monoid $M$, determines whether $M\in{\mathcal F}$.   The definition extends in the obvious fashion to families of ordered monoids and stamps.  For ordered monoids the input includes, in addition to the multiplication table of $M$, a representation of the graph of the partial order on $M$.  For stamps $\phi\colon A^*\to M$ we are also given the values $\phi(a)$ for $a\in A$.

We will say that a variety ${\mathcal V}$ of languages is decidable if $A^*{\mathcal V}$ is decidable for every finite alphabet $A$.  In this case the Eilenberg correspondence theorem gives a rather obvious connection between the two kinds of decidable families:

\begin{theorem}
A (positive) variety (respectively, ${\mathcal C}$-variety)  of languages is decidable if and only if the corresponding pseudovariety of (ordered) monoids (respectively, stamps) is decidable.
\end{theorem}

\begin{proof}
We give the proof just for the case of ordinary varieties of languages and pseudovarieties of monoids; the argument is essentially the same for all the other variants. Let ${\mathcal V}$ be a variety of languages and $\psvV$ the corresponding pseudovariety of monoids.  Suppose first that $\psvV$ is decidable. Let ${\mathcal A}=(Q,A,i,F)$ be a DFA recognizing a language $L\subseteq A^*$.  From ${\mathcal A}$ we can effectively construct the multiplication table of $\synt(L)$.  We then apply the algorithm for $\psvV$ to decide whether $\synt(L)\in \psvV$, and thus whether $L\in A^*{\mathcal V}$.  Conversely, suppose ${\mathcal V}$ is decidable.  Let $M$ be a finite monoid and choose a finite alphabet $A$ together with a surjective morphism $\phi\colon A^*\to M$.  (For example, we could choose $A=M$ and $\phi$ the extension to $A^*$ of the identity map on $M$.)  Then by Lemma~\ref{SW:lemma: syntactic generation} and Corollary~\ref{SW:cor: syntactic generation}, $M$ divides the direct product of the monoids $\synt(\phi^{-1}(m))$ for $m\in M$, and each of the $\synt(\phi^{-1}(m))$ in turn divides $M$.  Thus $M\in \psvV$ if and only if each of the languages $\phi^{-1}(m)$ is in $A^*{\mathcal V}$.  Furthermore, from $\phi$ we can construct a DFA $(M,A,1,\{m\})$ recognizing $\phi^{-1}(m)$, and thus decide whether each is in $A^*{\mathcal V}$.  Thus $\psvV$ is decidable.
\end{proof}

Decision problems for varieties of regular languages can have arbitrarily large computational complexity, or indeed be undecidable.  To see this, observe simply that if $P$ is \emph{any} set of primes, then we can form the pseudovariety $\psvG_P$ of finite groups $G$ such that every prime divisor of $|G|$ is in $P$.  Testing membership of a given prime $p$ in $P$ then reduces, in time polynomial in $p$, to testing membership in $\psvG_P$, so $\psvG_P$ is at least as complex as $P$.

On the other hand, Reiterman's theorem, which says varieties are defined by sets of profinite identities, suggests that we could determine membership in varieties simply by verifying whether identities hold in finite monoids.  This is deceptive, since elements of $\widehat{X^*}$ do not generally have simple descriptions that make it possible to evaluate their images in finite monoids, and, further, the equational description of a pseudovariety might require inifinitely many profinite identities.  We can nonetheless say something definitive about the complexity of the decision problems in the case where the equational definition consists of a finite set of profinite identities $\rho=\sigma$, where $\rho$ and $\sigma$ are $\omega$-terms in $\widehat{X^*}$:  This means that $\rho$ and $\sigma$ are formed from elements of $X$ by successive application of concatenation and the operation $\tau\mapsto\tau^{\omega}$.

\begin{theorem} Let ${\mathcal V}$ be a variety of languages defined by a finite set of profinite identities of the form $\rho=\sigma$, where $\rho$ and $\sigma$ are $\omega$-terms, and let $\psvV$ be the corresponding pseudovariety of finite monoids.  Then $\psvV$ is decidable by a logspace algorithm in the size of the input multiplication table, and ${\mathcal V}$ is decidable by a polynomial space algorithm in the size of the input automaton.

\end{theorem}

\begin{proof} We first consider testing membership of a monoid $M$ in $\psvV$.
Let $|M|=n$.  The multiplication table of $M$ can be represented in $O(n^2\log n)$ bits and each element of $M$ by $O(\log n)$ bits.  We will show how to determine membership of $M$ in $\psvV$ using $k\cdot \log_2 n$ additional bits of workspace, where the constant $k$ is determined by the length of the longest $\omega$-term occurring in the defining profinite identities for $\psvV$.  To make the proof easier to follow, let us suppose we have an identity $((x^{\omega} y)^{\omega} z)^{\omega}= (xz)^{\omega}$.  The algorithm loops through all triples $(x,y,z)$ of elements of $M$ and writes them in the workspace.  It then uses $\log_2 n$ bits of additional workspace to compute $x^{\omega}$.  This is done by repeatedly consulting the multiplication table, writing $x^2, x^3, \ldots $ in the same workspace, and after each write, consulting the multiplication table to check if the element is idempotent.  We similarly compute $(x^{\omega}y)^{\omega}$,  $((x^{\omega} y)^{\omega} z)^{\omega}$, and $(xyz)^{\omega}$. All in all, we used $7\cdot \log_2 n$ bits of workspace.  After all the values are computed, we compare the last two.  The algorithm rejects if it finds a mismatch.  If it finds none, it goes on to the next identity, and accepts if all the identities are tested with no mismatch.

We now turn to testing membership in ${\mathcal V}$. The algorithm we give is actually a nondeterministic polynomial space algorithm for nonmembership of a regular language in $A^*{\mathcal V}$.  Since, by Savitch's Theorem (~\cite{Savitch:1970}, see, also Sipser~\cite{Sipser:2006}) nondeterministic polynomial space is equivalent to deterministic polynomial space, and the latter is closed under complement, this will be enough. Let us work with the same example identity we used in the first part of the proof.  The algorithm begins by guessing words $x,y,z$ and computing the vectors 
\begin{align*}
&(q_1x,\ldots, q_nx),\\
&(q_1y,\ldots, q_ny),\\
&(q_1z,\ldots, q_nz),
\end{align*}
where $\{q_1,\ldots,q_n\}$ is the set of states of the input DFA.  Observe that the words $x, y, z$ themselves are not stored.  Instead they are guessed letter by letter, and only the vectors of states are written in the workspace. This requires $O(n\log n)$ bits, where $n$ is the number of states of the DFA. Observe as well that once we have the vector $(q_1u,\ldots, q_nu)$ we can, with an additional $n\log_2 n$ bits, compute the vector $(q_1u^{\omega},\ldots, q_nu^{\omega})$, since we can write the vectors of the successive powers $(q_1u^k,\ldots, q_nu^k)$ reusing the same workspace, and then check after each write whether $qu^k=qu^{2k}$ for each state $q$.  As a result we obtain the vectors $(q_1\hat\phi(\rho),\ldots,q_n\hat\phi(\rho))$, $(q_1\hat\phi(\sigma),\ldots,q_n\hat\phi(\sigma))$ for some morphism $\phi\colon X^*\to A^*$.  If these vectors turn out to be different, we accept.  Thus this algorithm nondeterministically recognizes the complement of $A^*{\mathcal V}$, using $O(n\log n)$ space.
 \end{proof}
 
 The foregoing theorem illustrates a potentially large gap in complexity between testing membership in ${\mathcal V}$ from an input DFA and testing membership in the corresponding pseudovariety $\psvV$ from the multiplication table of a monoid.  This is to be expected, since an automaton is in general exponentially more succinct than the multiplication table of its transition monoid.  In some instances, however, it is possible to give efficient algorithms that begin with automata, using so-called `forbidden pattern'\index{forbidden pattern} characterizations of varieties.  We illustrate this with a very simple example, using the ordered variety ${\mathcal J}^+$.  Consider the following figure:
 
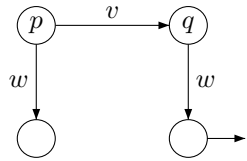
\begin{center}
\begin{picture}(25,22)(0,-22)
\node[NLangle=0.0,Nw=5.0,Nh=5.0,Nmr=2.5](n0)(0.0,-4.0){$p$}

\node[NLangle=0.0,Nw=5.0,Nh=5.0,Nmr=2.5](n1)(20.0,-4.0){$q$}

\node[NLangle=0.0,Nw=5.0,Nh=5.0,Nmr=2.5](n2)(0.0,-19.0){}

\node[NLangle=0.0,Nmarks=f,Nw=5.0,Nh=5.0,Nmr=2.5](n3)(20.0,-19.0){}

\drawedge(n0,n1){$v$}

\drawedge[ELside=r](n0,n2){$w$}

\drawedge(n1,n3){$w$}

\end{picture}
\end{center}
 
 We say that a DFA $(Q,A,i,F)$ contains this pattern if there are states $q_1, q_2$ and words $u,v,w\in A^*$ such that $iu=q_1$, $q_2=q_1v$, $q_1w\in F$, $q_2w\notin F$.  We say the DFA avoids the pattern if it does not contain it.  It is easy to see that a DFA recognizing a language $L$ avoids this pattern if and only if whenever $uw\in L$, $uvw\in L$.  Thus the languages in $A^*$ avoiding the pattern are exactly those that satisfy the inequality $x\leq 1$; that is, the language family $A^*{\mathcal J}^+$. We use this to prove the following:
 
 \begin{theorem} There is an algorithm determining membership in ${\mathcal J}^+$ that runs in nondeterministic logspace in the size of an accepting DFA. (In particular, membership can be determined in polynomial time.)
 
 \end{theorem}
 
 \begin{proof} 
 We nondeterministically guess letters to obtain an accessible state $q_1$, using $\log_2 n$ bits, where $n$ is the number of states in the automaton.  We then further guess letters to obtain another state $q_2=q_1v$, written on another $\log_2 n$-bit field in the work space.  Finally, we guess more letters, applying them to both components of the pair $(q_1,q_2)$ and arrive at at a state $(q_1w,q_2w)$.  We accept if the first member of this pair of states is an accepting state of the DFA and the second is not.  Thus we have a nondeterministic logspace algorithm for the regular languages outside of ${\mathcal J}^+$.  But by the theorem of Immerman and Szelepcsenyi (see \cite{Immerman:1988}, \cite{Szelepcsenyi:1988}, also \cite{Sipser:2006}), nondeterministic logspace is closed under complement, so we have the desired result. 
 
 \end{proof}
 
 The same reasoning is used in many proofs showing that varieties of languages are decidable in nondeterministic logspace:  find a forbidden pattern characterization of the variety using a fixed number of states. (For instance, Pin and Weil \cite{Pin&Weil:1997}, Glasser and Schmitz~\cite{Glasser&Schmitz:2008}.) While such results appear to bridge the complexity gap between poly\-nomial-time algorithms that begin with a multiplication table and exponential-time algorithms that begin with an automaton, forbidden pattern arguments are not always available.  In particular, we have the following result, which we cite without proof, from Cho and Huynh ~\cite{Cho&Huynh:1991}:
 
 \begin{theorem}
 Testing whether a regular language given by a DFA is aperiodic is PSPACE-complete.
 \end{theorem}

\section{Connections with logic}\label{SW:sec: logic}

In Section 1 we outlined, in an informal way, some of the logical apparatus for expressing properties of words over a finite alphabet.  Here we give a more precise and general description.  As before, variable symbols $x, y, x_1, x_2$, \emph{etc.,} denote positions in a word.  For each $a\in A$ our logics have a unary predicate symbol $Q_a$, where $Q_ax$ is interpreted to mean `the symbol in position $x$ is $a$.'  We also have a binary predicate symbol $s$,  where $s(x,y)$ is interpreted to mean `position $y$ is the successor of position $x$'.  We will usually use the alternative notation $y=x+1$ for this.  

We now consider \emph{monadic second-order}\index{formula!monadic second-order} formulas over this base of predicates.  These are formulas built not merely by quantifying over individual positions, but also by quantifying over \emph{sets} of positions, denoted by upper-case variable letters, and employing an additional relation symbol $x\in X$ between positions (first-order variables) and sets of positions (second-order variables).

For example, consider the monadic second order formula $\phi$:
\begin{displaymath}
\exists x\exists y\exists X(Q_ax \wedge Q_by \wedge x\in X \wedge y\in X\wedge\phi_1\wedge\phi_2),
\end{displaymath}
where $\phi_1$ is
\begin{displaymath}
\neg\exists z(x=z+1\wedge z\in X)\wedge\neg\exists z(z=y+1\wedge z\in X),
\end{displaymath}
and $\phi_2$ is 
\begin{displaymath}
\forall z(z\in X\rightarrow (y=z\vee \exists u(u\in X\wedge u=z+1)).
\end{displaymath}
The formula $\phi$ is a \emph{sentence}\index{sentence}; that is, it has no free variables. Thus $\phi$ defines a language $L_{\phi}$ over $A=\{a,b\}$, namely the set of all words in which the formula is true. The sentence asserts the existence of positions $x$ and $y$ with letters $a$ and $b$ respectively, and of a set $X$ of positions that contains both $x$ and $y$, that contains the successor of each of its elements with the exception of $y$, and that contains no elements less than $x$. Thus $L_{\phi}$ is the regular language $A^*aA^*bA^*$.

This example is an instance of the following important theorem, due to J. R. B\"uchi~\cite{Buchi:1960} (see \cite{Libkin:2004,Straubing:1994}).

\begin{theorem}
A language $L\subseteq A^*$ is regular if and only if $L=L_{\phi}$ for some sentence $\phi$ of monadic second-order logic.
\end{theorem} 

We obtain subclasses of regular languages by restricting these second-order formulas in various ways.  One obvious such restriction is to study first-order formulas\index{formula!first-order}: those formulas that use no second-order quantification.  We denote this logic, as well as the family of regular languages that can be defined in it, by $\FO[+1]$.  More generally, consider any $k$-ary relation $\alpha$ on the set of positions in a word that does not depend on the letters that appear in the word.  Suppose further that $\alpha(x_1,\ldots, x_k)$ is definable by a formula of monadic second-order logic.  Then we obtain a subclass of the regular languages by considering those languages definable by first-order sentences in which $\alpha$ is allowed as an atomic formula. We denote this class $\FO[\alpha]$,  and similarly write $F[\alpha_1,\alpha_2,\ldots]$ when there are several such predicates. For example, the relation $x<y$ is definable in monadic second-order logic, by a formula much like the one used above to define the language $L=A^*aA^*bA^*$. Thus we obtain the logic and the language class $\FO[<]$.  Of course, $L$ is definable in this logic, by the very simple sentence 
\begin{displaymath}
\exists x\exists y(Q_ax\wedge Q_by\wedge x<y).
\end{displaymath}
We can extend the expressive power further, by adjoining, for  $k>1$,  a binary predicate $\equiv_k$ that says two  positions are equivalent modulo $k$. These predicates, too, are definable in monadic second-order logic, and thus we obtain language classes $\FO[<,\equiv_k]$.  We can further restrict these families by bounding the quantifier depth, or the alternation of existential and universal quantifiers, or the number of distinct variable symbols.

We are interested in understanding the expressive power of these logics, and determining exactly what languages can be defined in them.  The critical insight is that, essentially, \emph{(nearly) all these language classes are varieties.}  In some instances we obtain ordered varieties, in others ${\mathcal C}$-varieties for a class ${\mathcal C}$ of morphisms, but in all cases we obtain families that, at least in principle, admit  a characterizations in terms of the syntactic monoids and morphisms of the languages they contain.

\subsection{Model-theoretic games}

To see why this is so, we first describe an important tool for studying the expressive power of logics for words.  Consider a first-order logic $\FO[\alpha_1,\ldots,\alpha_m]$.
Look at a pair of words $w,w'\in A^*$ and suppose that on each word we have placed $k$ `pebbles' labeled $x_1,\ldots,x_k$ for $w$, and $x_1',\ldots, x_k'$ for $w'$.  Each pebble is placed on a single position in its word, but two different pebbles can be on the same position.  We denote the resulting pebbled words by $u=(w,x_1,\ldots,x_k)$ and $u'=(w,x_1',\ldots, x_k')$.

We will now describe a game ${\mathcal G}_r(u,u',\alpha_1,\ldots,\alpha_k)$ played on these two pebbled words. (This is called an \emph{Ehrenfeucht Fra\"{\i}ss\'e game}\index{Ehrenfeucht-Fra\"\i ss\'e game}.) The subscript $r$ denotes the number of rounds of the game.  There are two players, traditionally called \emph{Spoiler,} who plays first, and \emph{Duplicator} who plays second.  We define the rules of the game by induction on the number of rounds.  In the 0-round game, the winner is already determined:  If there is a relation $\alpha=\alpha_i$ of arity $p$, and pebbles $x_{i_1},\ldots, x_{i_p}$,
$x'_{i_1},\ldots, x'_{i_p}$, such that 
\begin{displaymath}
\alpha(x_{i_1},\ldots,x_{i_p})
\end{displaymath}
holds, and
\begin{displaymath}
\alpha(x'_{i_1},\ldots,x'_{i_p})
\end{displaymath}
does not, or vice-versa, then Spoiler wins the game.  If there are pebbles $x_i$ and $x'_i$ such that the letter in position $x_i$ of $w$ is different from the letter in position $x'_i$ of $w'$, then Spoiler also wins the game.  Otherwise, Duplicator wins.  The idea is that Spoiler wins if the two pebbled words are different, and the difference must be witnessed by the atomic formulas applied to the pebbled positions.

Now let $r>0$.  In the $r$-round game ${\mathcal G}_{r}(u,u',\alpha_1,\ldots,\alpha_m)$, Spoiler makes a play by placing a new pebble $x_{k+1}$ in $u$ or $x'_{k+1}$ in $u'$.  If Spoiler played in $u$ then Duplicator must respond with $x'_{k+1}$ in $u'$.  Otherwise Duplicator responds with $x_{k+1}$ in $u$.  The result is two new pebbled words $v,v'$.  Spoiler and Duplicator proceed to play the game ${\mathcal G}_{r-1}(v,v',\alpha_1,\ldots,\alpha_m)$.  Whoever wins this $(r-1)$-round game is the winner of the $r$-round game.

Ordinary words  may be considered as special instances of pebbled words and thus we can consider the games ${\mathcal G}_r(w,w',\alpha_1,\ldots,\alpha_m)$, where $w,w'\in A^*$. The fundamental property of such games is given by the following theorem.

\begin{theorem}~\label{SW:thm: efgames}
Let $w,w'\in A^*$, $r\geq 0$. The words $w$ and $w'$ satisfy the same sentences in $\FO[\alpha_1,\ldots,\alpha_m]$ of quantifier depth $r$ or less if and only if  Duplicator has a winning strategy in ${\mathcal G}_r(w,w',\alpha_1,\ldots,\alpha_m)$.
\end{theorem}

See, for example, \cite{Libkin:2004,Straubing:1994}.

Here is an example:  Consider the two words $w=aab$ and $w'=aaab$. Spoiler has a winning strategy if ${\mathcal G}_2(w,w',<)$:  First play pebble $x_1$ on the second $a$ of $w'$.  If Duplicator replies on the first $a$ of $w$, Spoiler will play $x_2$ on the first $a$ of $w'$.  If Duplicator instead replies on the second $a$ of $w$, then Spoiler plays $x_2$ on the third $a$ of $w'$.  In either case, Duplicator has nowhere to play $x_2'$ in $w$ and win the game.  By Theorem~\ref{SW:thm: efgames}, there must be some sentence of quantifer depth 2 that distinguishes the two words.  Indeed, $w'$ satisfies
\begin{displaymath}
\exists x(Q_ax \wedge \exists y(Q_a y \wedge x<y)\wedge\exists y(Q_ay \wedge y<x)),
\end{displaymath}
while $w$ does not.  On the other hand, Duplicator has a winning strategy in the two-round game in $aaaab, aaab$.

What does this have to do with varieties?  We will use games to show that logically-defined language classes satisfy the closure properties that define varieties. Look, for example, at the family of languages defined by $\FO[<]$ sentences of quantifier depth no more than $d$, where $d\geq 0$. We will denote both this language family and the underlying logic by $\FO_d[<]$.  

\begin{theorem}\label{SW:thm: fo_variety}
$\FO_d[<]$ is a variety of languages.
\end{theorem}

\begin{proof}
Since we have to discuss languages over different alphabets, let us denote by $A^*\FO_d[<]$ the languages over $A^*$ that belong to this family.  Obviously $A^*\FO_d[<]$ is closed under boolean operations, so we must verify closure under quotients and inverse images of morphisms.  Let us write $w\sim_{d,A} w'$ to mean that $w,w'\in A^*$ satisfy all the same sentences of $\FO_d[<]$.  Then $\sim_d$ is an equivalence relation of finite index on $A^*$, and every language of $A^*\FO_d[<]$ is a union of $\sim_{d,A}$-classes. We claim that if $w\sim_{d,A}w'$ and $a\in A$, then both $aw\sim_{d,A}aw'$, and $wa\sim_{d,A}w'a$, and that further, if $\phi\colon A^*\to B^*$ is a morphism, then $\phi(w)\sim_{d,B}\phi(w')$.

To see that this claim implies the result, suppose $L\in A^*\FO_d[<]$ but $a^{-1}L\notin A^*\FO_d[<]$.  Then there exist $w,w'\in A^*$ with $w\in a^{-1}L$, $w'\notin a^{-1}L$, and $w\sim_{d,A} w'$.  But then $wa \in L$, $w'a\notin L$, and $wa\sim_{d,A}w'a$, contradicting $L\in A^*\FO_d[<]$.  By the same reasoning we deduce closure under right quotients and under inverse images of morphisms.

To prove the claim, note that by Theorem~\ref{SW:thm: efgames}, $w\sim_{d,A}w'$ if and only if Duplicator has a winning strategy in ${\mathcal G}_d(w,w',<)$.  So we must show that such a winning strategy implies the existence of winning strategies for Duplicator in 
${\mathcal G}_d(aw,aw',<)$, ${\mathcal G}_d(wa,w'a,<)$, and ${\mathcal G}_d(\phi(w),\phi(w'),<)$.
For ${\mathcal G}_d(wa,w'a,<)$, the strategy is this: Whenever Spoiler plays on the last letter of either $wa$ or $w'a$, Duplicator responds by playing on the last letter of the other word; otherwise Duplicator responds according to the winning strategy in $(w,w')$. The reasoning is identical for ${\mathcal G}_d(aw,aw',<)$. For ${\mathcal G}_d(\phi(w),\phi(w'),<)$, suppose $w=a_1\cdots a_r$, $w'=a_1'\cdots a_s'$, and let $v_i=\phi(a_i)$, $v_i'=\phi(a_i')$.  Duplicator's strategy is to keep track of a separate game in $w, w'$ to calculate the responses in $\phi(w), \phi(w')$. If Spoiler plays on the $j^{th}$ symbol of $v_i$, then Duplicator calculates the response, according to the original strategy, to a move by Spoiler on $a_i$.  Let us say this response is on $a_k'$.  Observe that $a_i=a_k'$, and thus $v_i=v_k'$, so Duplicator can reply on the $j^{th}$ symbol of $v_k'$. In other words, Duplicator pulls the Spoiler's plays back to $(w,w')$, applies the original winning strategy, and pushes the result forward to $(\phi(w),\phi(w'))$. It is easy to see that this strategy wins for Duplicator.
\end{proof}

This same reasoning can be adapted to a large number of different situations.  Consider, for example, the logics $\FO_d[+1]$.  The strategy-copying argument no longer works to give Duplicator a winning strategy in ${\mathcal G}_d(\phi(w),\phi(w'),+1)$, because $\phi$ may map a letter to the empty word, and thus we might end up with two pebbles on adjacent positions in $\phi(w)$, but find the corresponding pebbles on non-adjacent positions of $\phi(w')$.  But the argument \emph{does} work for non-erasing morphisms, and thus each $\FO_d[+1]$, as well as the union $\FO[+1]$, is a ${\mathcal C}_{ne}$-variety. Similarly, suppose we augment the logic $\FO[<]$ by adjoining the predicate  $x\equiv_q y$ for equivalence modulo $q$.  We now find that the strategy-copying argument works as long as all $\phi(a)$ for $a\in A$ have the same length $m$, as
$i\equiv_q j$ implies $mi\equiv_q mj$. Thus each $\FO_d[<,\equiv_q]$ is a ${\mathcal C}_{lm}$-variety of languages.

This reasoning is amenable to further adaptations, by altering the rules of the games:  We obtain a game characterization of languages defined by formulas that use no more than $p$ distinct variables by allowing only $p$ pebbles, regardless of the number of rounds.  Once all the pebbles have been placed, the Spoiler may pick up a pebble and move it to a new position; the Duplicator must pick up the corresponding pebble and move it in the same direction. We obtain a game characterization of the languages defined by boolean combinations of $\Sigma_k$ sentences\footnote{Formulas in prenex normal form with at most $k-1$ alternating blocks of quantifiers, or with exactly $k$ blocks where the first block is existential.}, with quantifier block size bounded by $d$, by considering $k$-round games in which each player is permitted to place $d$ pebbles at a time.  We can turn this into a  game characterization of the languages defined by $\Sigma_k$-sentences themselves by requiring Spoiler to play in $w$ in the first round, in $w'$ in the second round, \emph{etc.} Duplicator then has a winning strategy  in the game in $w,w'$ if and only if every $\Sigma_k$-sentence, with quantifier block size no more than $d$, that $w$ satisfies is also satisfied by $w'$. We can use this to conclude that $\Sigma_k[<]$ is an ordered variety of languages.  In all instances, we find that some variant of Eilenberg's Theorem applies, and extract the same conclusion:  A logical characterization of the language class implies the existence of an algebraic characterization.

Care must be taken not to extrapolate this \emph{too} far.  For example, the strategy-copying argument fails in the case of $\Sigma_1[+1]$:  Let $w=abab$, $w'=baba$.  Then $w,w'$ satisfy the same $\Sigma_1[+1]$-sentences of block size 2, but $wa$ and $w'a$ do not, since $w'a$ contains two consecutive occurrences of $a$. 
   
\subsection{Explicit characterization of logically defined classes}

While the foregoing arguments tell us that logically defined language classes form varieties, they do not provide explicit algebraic characterizations.  There are, in fact, a number of different methods for connecting the structure of defining sentences to algebraic properties, and many results giving explicit characterizations of the language varieties defined by various logics. (See, for instance Straubing~\cite{Straubing:1994}.)  Here we give just a taste of these techniques and results with what is perhaps the most famous, and certainly the first, result in this area, the theorem of McNaughton and Papert\index{McNaughton's and Papert's Theorem}~\cite{McNaughton&Papert:1971} giving the equivalence of first-order logic and aperiodic monoids:

 \begin{theorem}\label{SW:thm:fo_aperiodic}
 A language $L$ belongs to $\FO[<]$ if and only if $\synt(L)$ is aperiodic.
 \end{theorem}

We will only prove one direction of this theorem, namely that first-order definability implies aperiodicity.  We claim that if $u\in A^*$, then $u^{2^d-1}\sim_{d,A} u^{2^d}$.  This is proved by induction on $d$.  For $d=0$, there is nothing to prove, since all words are equivalent modulo $\sim_{0,A}$.  Suppose then that $d>0$.  We will show that Duplicator has a winning strategy in ${\mathcal G}_d(u^{2^d-1},u^{2^d},<)$.  Suppose Spoiler plays $x_1$ in $u^{2^d-1}$. 
\begin{displaymath}
u^{2^d-1}=u^rvav'u^s,
\end{displaymath}
where the pebble is played on the position indicated by the letter $a$, $u=vav'$, and $r+s=2^d-2$.  It follows that either $r\geq 2^{d-1}-1$ or $s\geq 2^{d-1}-1$.  Suppose the former (the proof is the same in either case).  Then we can write
\begin{displaymath}
u^{2^d}=u^{r+1}vav'u^s.
\end{displaymath}
Duplicator places the pebble $x_1'$ on the indicated $a$.  Now play proceeds as follows:  By the inductive hypothesis, Duplicator has a winning strategy in ${\mathcal G}_{d-1}(u^{2^{d-1}-1},u^{2^{d-1}})$. Thus, by the argument given in the proof of Theorem~\ref{SW:thm: fo_variety}, Duplicator has a winning strategy in ${\mathcal G}_{d-1}(u^{r}v,u^{r+1}v)$.  Duplicator will follow this strategy whenever Spoiler plays to the left of $x_1$ or $x_1'$, and simply copy Spoiler's move in $av'u^s$ whenever the play is at or to the right of $x_1$ or $x_1'$.  This proves the claim.  It follows that if $L$ is first-order definable, then $\synt(L)$ satisfies the $x^m=x^{m+1}$ for sufficiently large $m$, and is thus aperiodic.

We omit the proof of the converse, that if $\synt(L)$ is aperiodic, then $L$ is in $\FO[<]$.  Most of the published proofs of this theorem rely on some decomposition theory for finite semigroups, either the Krohn-Rhodes decomposition, or the ideal structure of semigroups. Most proofs also show first that every language recognized by an aperiodic monoid is a star-free language.  We will define star-free languages in Section~\ref{SW:sec: malcev}, and show that they are equivalent to first-order definable languages. Pin \cite{Pin:1986} gives a relatively streamlined proof using the ideal decomposition theory. Straubing \cite{Straubing:1994} uses the Krohn-Rhodes decomposition to obtain a first-order sentence directly.  Wilke \cite{Wilke:1999} gives a proof that is remarkable for its absence of hard semigroup theory, and that produces a formula of temporal logic directly from an automaton with an aperiodic transition monoid.
\cqfd

We can use Theorem~\ref{SW:thm:fo_aperiodic} to deduce a claim we made earlier, giving an explicit characterization of the ${\mathcal C}_{lm}$-pseudovariety {\bf QA}:

\begin{theorem}\label{SW:thm: quasiaperiodic}
$L$ belongs to $\FO[<,\equiv_m]$ for some $m>1$ if and only if the syntactic morphism of $L$ is in ${\bf QA}$. 
\end{theorem}

We merely sketch the argument: Suppose $u\in A^+$ with $|u|$ divisible by $m$.  Let $d>0$. Then by precisely the same argument as we gave in the proof of Theorem~\ref{SW:thm:fo_aperiodic}, Duplicator has a winning strategy in ${\mathcal G}_d(u^r,u^{r+1},<,\equiv_m)$ as long as $r$ is sufficiently large compared to $d$. This is enough to show that if $L$ is definable by a sentence of $\FO[<,\equiv_m]$, then the stable semigroup of $\eta_L$ is aperiodic.  For the converse, we consider a language $L$ with $\eta_L$ in {\bf QA}.  Let $\eta_L(A^t)$ be the stable semigroup.  If we treat $B=A^t$ as a finite alphabet, we can use Theorem~\ref{SW:thm:fo_aperiodic} to obtain a first-order sentence, with respect to $B$, defining the sets of words of length divisible by $t$ that are recognized by $\eta_L$, and then translate this to a first-order sentence over $A$ by means of the predicate $\equiv_t$.
\cqfd

\subsubsection*{Other logical formalisms}

By and large, we have confined our discussion of logic to the use of first-order quantification.  But there are other formalisms studied in the literature, which also give rise to varieties. We mention in passing two of these: Formulas with modular quantifiers, which were introduced by Straubing, Th\'erien and Thomas~\cite{Straubing&Therien&Thomas:1995} and studied extensively in~\cite{Straubing:1994}, and \emph{temporal} formulas, which play an important role in computer-aided verification.  An algebraic treatment of temporal logic, and its connection to varieties of languages, is due to Th\'erien and Wilke~\cite{Therien&Wilke:1998,Therien&Wilke:2001} and Wilke~\cite{Wilke:1999,Wilke:2001}.

Considerable effort has been devoted to the effective characterization of particular fragments of $\FO[<]$.  One approach is based on restricting the number of bound variables appearing in a formula.  Consider, for instance the sentence
$$\exists x(Q_ax\wedge \exists y(Q_by\wedge y=x+1 \wedge \exists x(Q_bx\wedge x=y+1))),$$
which defines the set of strings that contain $abb$ as a factor.  While the quantifiers in this sentence are nested three levels deep, only two variable symbols are used, because we were able to re-use the symbol $x$.  Immerman and Kozen~\cite{Immerman&Kozen:1989} showed that {\it any} sentence of $\FO[<]$ can be rewritten as an equivalent sentence that uses only three variables.  We write this result as 
$$\FO[<]=\FO^3[<].$$
The question that naturally arises is what one can do with {\it two} variables---that is, what is the expressive power of $\FO^2[<]$?  It is known that the inclusion of $\FO^2[<]$ in $\FO[<]$ is strict; for example, it is not hard to show that the language $(ab)^*$ cannot be defined by a formula with fewer than three variables.  The exact answer turned out to be quite interesting:  A language $L$ is definable in $\FO^2[<]$ if and only if $\synt(L)$ belongs to the pseudovariety $\psvDA$ defined by the equations
$$\llbrack (xyz)^\omega z (xyz)^\omega = (xyz)^\omega\rrbrack.$$
This variety had been discovered much earlier by Sch\"utzenberger~\cite{Schutzenberger:1976} and arises in many different contexts (see, for example, Tesson and Th\'erien~\cite{Tesson&Therien:2002} and the discussion in Section~\ref{SW:sec: malcev}).
This opened a rich vein of related research on varieties defined by two-variable logics (\textit{e.g.} Kufleitner and Weil~\cite{Kufleitner&Weil:2012}, Krebs and Straubing~\cite{Krebs&Straubing:2012,Krebs&Straubing:2017}, Kufleitner and Lauser \cite{Kufleitner&Lauser:2013}, Fleischer, Kufleitner and Lauser \cite{Fleischer&Kufleitner&Lauser:2017}, Krebs {\it et al.}~\cite{Krebs&Lodaya&Pandya&Straubing:2016}).

We have already alluded to the fragments $\Sigma_k[<]$ and $\mathcal{B}\Sigma_k[<]$ (boolean combinations of $\Sigma_k$ sentences). These all give varieties of languages (ordered varieties in the case of $\Sigma_k[<]$), but effective characterization of these varieties for all but the lowest levels ($\Sigma_1[<]$, $\mathcal{B}\Sigma_1[<]$, $\Sigma_2[<]$) has been an outstanding open problem.  Recently, Place and Zeitoun  ~\cite{Place:2015,Place&Zeitoun:2014,Place&Zeitoun:2015-siglog,Place&Zeitoun:2017,Place&Zeitoun:2017-csr}, made a critical breakthrough, developing a number of novel and difficult techniques for attacking this problem.  As a result, we now possess effective characterizations for the varieties of languages  $\mathcal{B}\Sigma_2[<]$, $\Sigma_3[<]$ and $\Sigma_4[<]$.

\subsubsection*{Beyond membership}

Much of what we have written concerning decision problems has focused on the membership problem for a pseudovariety $\psvV$:  Given a monoid $M$, determine whether it belongs to $\psvV$, or, equivalently, given a regular language $L$, determine whether $L$ belongs to the corresponding variety of languages $\mathcalV$.  In a series of papers (see the surveys \cite{Place&Zeitoun:2015-siglog,Place&Zeitoun:2017-csr}), Place and Zeitoun have embarked on a deep study of the {\it separation} problem for a variety of languages $\mathcalV$.  If $L_1,L_2\subseteq A^*$ are disjoint regular languages, we say that $L_1$ is {\it $\mathcalV$-separable} from $L_2$ if there exists a language $K\in A^*\mathcalV$ such that $L_1\subseteq K$ and $L_2\cap K=\emptyset$. \footnote{If $\mathcalV$ is a variety of languages corresponding to a pseudovariety of finite monoids, then this relation is symmetric: That is, $L_1$ is $\mathcalV$-separable from $L_2$ if and only if $L_2$ is $\mathcalV$-separable from $L_1$.  However, this is not the case for ordered varieties.  Observe that if we can decide $\mathcalV$-separability for pairs of languages, then we can decide the membership problem for $\mathcalV$, since this is just the question of separating $L$ from its complement.} This problem is equivalent to another one, expressed in algebraic and topological terms, namely the computation of \emph{$\psvV$-pointlike pairs} of a given monoid $M$. This was observed by Almeida in 1999 \cite{Almeida:1999}, but the problem of computing $\psvV$-pointlikes had been considered even earlier, notably in the difficult case where $\psvV = \psvA$ (Henckell \cite{Henckell:1988}, see also \cite{Henckell&Rhodes&Steinberg-a:2010,Henckell&Rhodes&Steinberg-b:2010} for a simpler proof, and \cite{vanGool&Steinberg:2018} for a generalization).
Place and Zeitoun's breakthrough results concerning the varieties $\Sigma_k[<]$ and $\mathcal{B}\Sigma_k[<]$ depend critically on the separation problem for varieties.

Further results on separation (or the computation of $\psvV$-pointlike sets) include transfer theorems such as those of Steinberg \cite{Steinberg-2:2001} on the computation of $\psvV\ast\psvD$ pointlikes, of Place and Zeitoun \cite{Place&Zeitoun-b:2017} on the preservation of separation by logical fragments when enriched with so-called local predicates (successor, min, max), or of Place, Ramanathan and Weil \cite{Place&Ramanathan&Weil:2018} on the enrichment of logical fragments with modular predicates.

\section{Operations on classes of languages}

The idea developed in this section is that certain operations on
classes of languages translate to operations on the corresponding sets
of profinite identities, or on the corresponding classes of syntactic
objects (syntactic monoids or semigroups, ordered or not, etc).  This
translation, when it can be made explicit, may provide decomposition
results, or membership decision results for complex classes of
languages.

\subsection{Boolean operations}

 If for each $i\in I$, $\calV_i$ is a class of regular languages, the
 intersection $\calW = \bigcap_{i\in I}\calV_i$ is the class given by
 $A^*\calW = \bigcap_{i\in I} A^*\calV_i$ for each alphabet $A$.  The
 different classes of families of languages considered so far (lattices
 or boolean algebras of languages of some fixed $A^*$, positive
 $\calC$-varieties) are easily seen to be closed under (arbitrary)
 intersection.  

 The following statement essentially follows from
 the definition of the satisfaction of profinite equations.

 \begin{proposition}
     Let $I$ be a set and for each $i\in I$, let $E_i$ be a set of
     profinite equations on an alphabet $A$.  Then $\bigcap_{i\in
     I}\calL(E_i) = \calL(\bigcup_{i\in I} E_i)$.
     
     In particular, if for each $i\in I$ $\calV_i$ is a class of
     regular languages that is $\calC$-defined by a set of profinite
     (ordered) $\calC$-identities $E_i$, then $\bigcap_{i\in I}
     \calV_i$ is $\calC$-defined by $\bigcup_{i\in I} E_i$.
 \end{proposition}
     
 The fact that an arbitrary intersection of lattices of regular
 languages (resp.  (positive) $\calC$-varieties) is again a lattice of
 regular languages (resp.  a (positive) $\calC$-variety) has the
 following consequence: for each set $V$ of regular languages in $A^*$
 (resp.  every class $\calV$ of regular languages) there exists a least
 lattice (resp.  a least (positive) $\calC$-variety) containing it,
 which is said to be \emph{generated by} $V$\index{lattice!generated by a set} (resp.  $\calV$\index{variety!generated by a set}).

 The union of two lattices of languages in $A^*$ is not a lattice in
 general.  The relevant operation is the \emph{join}: the join\index{join} of two
 lattices of regular languages in $A^*$ (resp.  classes of regular
 languages) is defined to be the lattice generated by their union.

 Describing the profinite equations or identities defining a join is
 difficult.  In fact, Albert, Baldinger and Rhodes exhibit
 \cite{Albert&Baldinger&Rhodes:1992} a finite set $\Sigma$ of computable profinite
 identities, such that the join of the pseudovariety
 $\llbrack\Sigma\rrbrack$ with the pseudovariety $\Com = \llbrack xy =
 yx\rrbrack$ of commutative monoids, is not decidable (see also \cite{Auinger&Steinberg:2003}).

 Some joins were computed early, based on the structural theory of
 monoids.  This is the case for instance of $\psvJ_1 \vee \psvG$, which is
 characterized as the class of finite monoids which are unions of groups and
 in which idempotents commute (see \cite{Howie:1976}).   This translates as
 \begin{displaymath}
 \psvJ_1\vee\psvG = \llbrack x^{\omega+1} = x,\ x^\omega y^\omega = 
 y^\omega x^\omega\rrbrack.
 \end{displaymath}
 Other joins resisted computation until the advent of profinite
 methods, such as the joins $\RR \vee \LL$ (Almeida and Azevedo
 \cite{Almeida&Azevedo:1989}) and $\psvG\vee \Com$ (Almeida
 \cite{Almeida:1988}).  The case of $\psvJ\vee\psvG$ is interesting, since
 this join is decidable but is not defined by a finite set of profinite
 identities (Almeida, Azevedo and Zeitoun \cite{Almeida&Azevedo&Zeitoun:1999},
 Steinberg \cite{Steinberg:1998,Steinberg:2001}, Trotter and
 Volkov \cite{Trotter&Volkov:1996}).

 \begin{example}
     The following simple examples will be useful in the sequel.  Let
     $\psvI = \llbrack x=y\rrbrack$ be the trivial pseudovariety of
     monoids (which consists only of the 1-element monoid).  Let
     $\psvK$ and $\psvD$ be, respectively, the pseudovarieties of semigroups
     $\psvK = \llbrack x^\omega y = x^\omega\rrbrack$ and $\psvD = \llbrack
     yx^\omega = x^\omega\rrbrack$.  The elements of $\psvK$ are the finite semigroups in which 
     idempotents act as zeroes on the left.  Dually, in the semigroups of 
     $\psvD$, idempotents act like zeroes on the right.  If $\psvV$ is any
     pseudovariety of monoids, we let $L\psvV$ be the class of finite
     semigroups $S$ such that $eSe \in \psvV$ for each idempotent $e$ of
     $S$.  It is easily verified that $L\psvV$ is a pseudovariety of
     semigroups, and that it is decidable if and only if $L\psvV$ is.
     
     It is also easy to verify that the semigroups that are both in $\psvK$
     and in $\psvD$ are exactly the semigroups with a single idempotent,
     which is a zero (these semigroups are
     called \emph{nilpotent}).
     Interestingly, the join $\psvK\vee\psvD$ is equal to $L\psvI = \llbrack
     x^\omega yx^\omega = x^\omega\rrbrack$.
%
 \end{example}

\subsection{Closure operations and Mal'cev products}\label{SW:sec: malcev}

An early closure result is  Sch\"utzenberger's theorem on star-free languages. The set of \emph{star-free languages}\index{language!star-free}\index{star-free languages} over an alphabet $A$ is the least boolean algebra containing the letters (and the empty set), which is closed under concatenation.  For instance, $aA^*$ is star-free, since it is equal to $a\emptyset^c$. A non-trivial question is that of decidability: given a regular language $L$, can we decide whether it is star-free? As it turns out, $(ab)^*$ is star-free (its complement is the set of all words with two consecutive $a$'s or two consecutive $b$'s, or that start with $b$ or end with $a$) but $(aa)^*$ is not.

The solution to this problem was given by Sch\"utzenberger \cite{Schutzenberger:1965} with the following theorem\index{Sch\"utzenberger's Theorem}.

\begin{theorem}\label{SW:thm: Schutz}
The class of star-free languages forms a variety of languages, corresponding to the pseudovariety $\psvA$ of aperiodic monoids. In particular, this class is decidable.
\end{theorem}

In view of Theorem~\ref{SW:thm:fo_aperiodic}, this is equivalent to the following statement.

\begin{theorem}\label{SW:thm: SF vs FO}
A language is star-free if and only if it is $\FO[<]$-definable.
\end{theorem}

\begin{proof}
We prove Theorem~\ref{SW:thm: SF vs FO} using game-theoretic methods, as in Section~\ref{SW:sec: logic}. Let us first show that a $\FO[<]$-definable language is star-free. It is sufficient to show, by induction on $k$, that for all $w\in A^*$ and $k\geq 0$, $[w]_k$ is star-free.  The case $k=0$ is trivial, since $[w]_0=A^*$ for all $w\in A^*$.  To prove the general case, we will establish the equality
\begin{displaymath}
[w]_{k+1}=\bigcap [x]_ka[x']_k\setminus \bigcup [y]_kb[y']_k,
\end{displaymath}
where the intersection is over all factorizations $w=xax'$ with $x,x'\in A^*$ and $a\in A$, and the union is over all triples $([y]_k,b,[y']_k)$, where $b\in A$ and $w\not \in [y]_kb[y']_k$.  By induction, the $\sim_k$-classes are star-free languages, so the equality above implies that the $\sim_{k+1}$-classes are star-free as well.

To prove the equality, note that the inclusion from left to right is trivial, so we need only show that if $w'\in A^*$ is in the set on the right-hand side, then $w\sim_{k+1} w'$.  So we will show that Duplicator has a winning strategy in the $(k+1)$-round game in the two words. Observe that inclusion of $w'$ in the right-hand side means that $w,w'$ have precisely the same set of factorizations with respect to $\sim_k$, in the sense that for every factorization $xax'$ of one word, with $a\in A$, there exists a corresponding factorization $yay'$ of the other word with $x\sim_kx'$, $y\sim_ky'$.  Thus if Spoiler plays on a position in one of the words, inducing a factorization $xax'$ of the word, Duplicator can play on the corresponding position of the other.  Duplicator can now correctly reply in the remaining $k$ rounds of the game by using her winning strategy in the games in $(x,y)$ and $(x',y')$.

Conversely, let us show that every star-free language is $\FO[<]$-definable. In view of the definition of star-free languages, we need to show, first, that $A^*$ and every language of the form $\{a\}$ ($a\in A$) is $\FO[<]$-definable; and second that if $K$ and $L$ are $\FO[<]$-definable, then so are the boolean combinations of $K$ and $L$, and so is $KL$. The only non-trivial point concerns the concatenation product, and the problem easily reduces to showing that $KaL$ ($a\in A$) is $\FO[<]$-definable.

Let us assume that $K$ and $L$ are defined by formulas of quantifier-depth $k$. Let $w \in KaL$, say, $w = uav$ with $u\in K$ and $v\in L$. We want to show that if $w \sim_{k+1} w'$ --- that is, Duplicator has a winning strategy for $\mathcal{G}_{k+1}(w,w')$ ---, then $w'\in KaL$. Let Spoiler put a pebble on the letter $a$ in $w$ witnessing the factorization $w = uav$, then Duplicator's strategy has her put a pebble on a letter $a$ in $w'$, determining a factorization $w' = u'av'$. We claim that Duplicator wins the $k$-round game in $u$ and $u'$: indeed, such a game can be seen as the 2nd, \dots, ($k+1$)-st moves in a game in $w = uav$ and $w' = u'av'$. Therefore $u\sim_k u'$ and hence $u'\in K$. Similarly $v'\in L$: thus $w' \in KaL$.
\end{proof}

A natural extension of the question answered by Sch\"utzenberger's theorem is the following:
 can we characterize the varieties of languages which are closed under
 concatenation product?  and if $\calV$ is a variety of languages, can
 we describe the least variety containing $\calV$ and closed under
 concatenation product? Both problems were solved by Straubing  \cite{Straubing:1979}. In order to state his 
 result, we need to introduce an operation on pseudovarieties.

Let $\psvV$ be a pseudovariety of monoids and let $\psvW$ be a pseudovariety
of semigroups (resp.  ordered semigroups).  We consider the class of
all finite monoids (resp.  ordered monoids) $M$ for which there exists
a morphism (un-ordered) $\phi\colon M \to N$ such that $N\in \psvV$ and
$\phi\inv(e)\in \psvW$ for each idempotent element $e$ of $N$.  This
class is not a pseudovariety in general, but it is elementary to
verify that the quotients (resp.  ordered quotients) of its elements
form a pseudovariety of monoids (resp.  ordered monoids), called the
\emph{Mal'cev product}\index{Mal'cev product} of $\psvV$ by $\psvW$, and denoted $\psvW\malcev\psvV$.

\begin{theorem}\label{SW:thm: concat closure}
    Let $\calV$ be a variety of languages and let $\psvV$ be the
    corresponding pseudovariety of monoids.  If $\calW$ is the least
    variety of languages containing $\calV$ and closed under
    concatenation product, then the corresponding pseudovariety of
    monoids is $\psvA\malcev\psvV$.
\end{theorem}

Sch\"utzenberger's
theorem above is the particular case of Theorem~\ref{SW:thm: concat closure}
when $\calV$ is the trivial variety of languages.

Interestingly, the Mal'cev product is also useful for characterizing the
closure of a variety of languages under other types of products.  For
technical reasons, the definition of these products involves
intermediate, marker letters: If $K$ and $L$ are languages in $A^*$,
and if $a\in A$, we say that the product $KaL$ is
\emph{deterministic}\index{product!deterministic} if each word $u\in KaL$ has a unique prefix in
$Ka$.  Co-deterministic products are defined dually: the product $KaL$
is \emph{co-deterministic} if each word $u\in KaL$ has a unique
suffix in $aL$.  Another important modality of product is the
following: a product $L_0a_1L_1\cdots a_kL_k$ is \emph{unambiguous}\index{product!unambiguous}
if every word $u$ in this language admits a unique decomposition in
the form $u = u_0a_1u_1 \cdots a_ku_k$ with each $u_i \in L_i$.
Deterministic and co-deterministic products are particular cases of
unambiguous products.

It is natural to extend these operations to classes of languages.
Given a class of languages $\calV$, we denote by $\Det\calV$ the class
of languages such that, for each alphabet $A$, $A^*\Det\calV$ is
the set of all boolean combinations of languages of $A^*\calV$ and of
deterministic products of these languages.  $\Det\calV$ is called the
\emph{deterministic closure}\index{closure!deterministic} of $\calV$.  The
\emph{co-deterministic closure} $\coDet\calV$ and the
\emph{unambiguous closure}\index{closure!unambiguous} $\UPol\calV$ are defined similarly.
Sch\"utzenberger \cite{Schutzenberger:1976,Pin:1986} characterized
algebraically these operations for varieties of languages.

\begin{theorem}\label{SW:thm: deterministic and unambiguous products}
    Let $\calV$ be a variety of languages and let $\psvV$ be the
    corresponding pseudovariety of monoids.  Then $\Det\calV$,
    $\coDet\calV$ and $\UPol\calV$ are varieties of languages,
    and the the corresponding pseudovarieties of monoids are
    $\psvK\malcev\psvV$, $\psvD\malcev\psvV$ and $L\psvI\malcev\psvV$, respectively.
\end{theorem}
 
\begin{example}\label{SW:example: R, L, DA}
    Consider the variety of languages $\calJ_1$, described in
    Sections~\ref{SW:subsection:j_1} and~\ref{SW:sec: j1 again}: for each alphabet $A$, $A^*\calJ_1$ is the
    boolean algebra generated by the languages of the form $B^*$, with
    $B \subseteq A$.  It is elementary to verify that $A^*\Det\calJ_1$
    is the boolean algebra generated by the products of the form
    $A_0^*a_1A_1^*\cdots a_kA_k^*$, such that for each $0 < i \le k$,
    $a_i\not\in A_{i-1}$.  Theorem~\ref{SW:thm: deterministic and
    unambiguous products} tells us that $\Det\calJ_1$ forms a variety 
    of languages, and that the corresponding pseudovariety of monoids 
    is $\psvK\malcev\psvJ_1$.
    
    Semigroup theory helps us characterize this pseudovariety.
    $\psvK\malcev\psvJ_1$ is the class $\RR$ of all so-called $\calR$-trivial
    finite monoids, that is, the monoids $M$ in which principal right
    ideals have a single generator: $sM = tM$ implies $s = t$.  In
    addition, one can show that $\RR = \llbrack (xy)^\omega x =
    (xy)^\omega\rrbrack$.  This immediately implies the decidability
    of $\Det\calJ_1$.
    
    A dual result characterizes $\psvD\malcev\psvJ_1$, the pseudovariety
    associated with $\coDet\calJ_1$, as the class $\LL$ of
    $\calL$-trivial finite monoids.  It is interesting to note that
    $\RR \cap \LL = \psvJ$. The variety of piecewise testable languages
    discussed in Section~\ref{SW:subsection:j} is therefore the class of
    languages that can be described simultaneously as boolean
    combinations of deterministic and of co-deterministic products of
    the form $A_0^*a_1A_1^*\cdots a_kA_k^*$ with each $A_i$ a subset
    of $A$.

    Similarly, Theorem~\ref{SW:thm: deterministic and unambiguous
    products} shows that the pseudovariety of monoids corresponding to
    $\UPol\calJ_1$ is is $L\psvI\malcev\psvJ_1$.  Again, one can show that
    this pseudovariety is the class of finite monoids in which every
    regular element is idempotent, usually denoted by $\psvDA$, and equal
    to $\llbrack (xyz)^\omega z (xyz)^\omega = (xyz)^\omega\rrbrack$.
    It follows, here too, that $\UPol\calJ_1$ is decidable.  Let us
    note in addition that it coincides with the class of languages
    that can be defined by $\FO[<]$ sentences that use at most two
    variable symbols. (See \cite{Tesson&Therien:2002}.)
\end{example}

The following result is of the same nature as Theorems~\ref{SW:thm:
concat closure} and~\ref{SW:thm: deterministic and unambiguous products}
but it involves a positive variety of languages, and the corresponding
pseudovariety of ordered monoids.
If $\calL$ is a set of regular languages in $A^*$, we denote by
$\Pol\calL$ (the \emph{polynomial closure}\index{closure!polynomial} of $\calL$), the lattice
generated by the languages of the form $L_0a_1L_1\cdots a_kL_k$, with
$L_i \in \calL$ and $a_i\in A$ for each $i$.  If $\calV$ is a class of
regular languages, then $\Pol\calV$ is the class such that, for each
alphabet $A$, $A^*\Pol\calV = \Pol(A^*\calV)$.  Then the following
result holds, see \cite{Pin&Weil:1996}.

\begin{theorem}\label{SW:thm: polV}
    Let $\calV$ be a variety of languages.  Then $\Pol\calV$ is a
    positive variety of languages, and the the corresponding
    pseudovariety of ordered monoids is $\llbrack x^\omega y
    x^\omega \le x^\omega\rrbrack \malcev\psvV$.
\end{theorem}
 


In general, the results reported above do not provide explicit
decision algorithms, even if $\psvV$ is decidable (see \cite{Auinger&Steinberg:2003}).  However, the
structural theory of semigroups yields some such results.  In
particular, we can use a result by Krohn, Rhodes and Tilson
\cite{Krohn&Rhodes&Tilson:1965} to show that if $\calV$ is decidable, then so
are $\Det\calV$, $\coDet\calV$ and $\UPol\calV$ (generalizing the 
specific instances discussed in Example~\ref{SW:example: R, L, DA}).

It is not known whether $\psvA\malcev\psvV$ is decidable whenever $\calV$ is. A positive
solution to this problem would imply a positive solution to an open instance of the complexity problem, which we discuss below in Section \ref{SW:subsec:products}.

Topological methods also \cite{Pin&Weil:1996} provide sets of
profinite identities describing Mal'cev products.  In the cases of
interest for us, it yields the following statement.

\begin{proposition}\label{SW:prop: equations for polV}
    Let $\calV$ be a variety of languages.  Then the least variety
    containing $\calV$ and closed under concatenation is defined by
    the set of profinite identities of the form $x^{\omega+1} =
    x^\omega$, where $x\in \widehat{X^*}$ and $\psvV$ satisfies $x =
    x^2$.
    
    Similar statements hold for $\Det\calV$ (respectively, $\coDet\calV$,
    $\UPol\calV$ and $\Pol\calV$), replacing the
    profinite identity $x^{\omega+1} = x^\omega$ by $x^\omega y =
    x^\omega$ (respectively, $yx^\omega = x^\omega$, $x^\omega yx^\omega = x^\omega$
    and $x^\omega y x^\omega \le x^\omega$),    
    where $x, y \in \widehat{X^*}$ and $\psvV$ satisfies $x =
    x^2 = y$.
\end{proposition}

These results were extended to $\calC$-varieties, and in the case of
$\Pol\psvV$, to lattices of regular languages closed under quotients
\cite{Pin&Straubing:2005,Branco&Pin:2009}.
In practice, the resulting sets of profinite identities are infinite 
and sometimes even uncomputable. However, in a number of situations, one can extract 
from these sets more manageable, yet sufficient subsets, yielding decision algorithms.

\begin{example}
Branco and Pin \cite{Branco&Pin:2009} use Proposition~\ref{SW:prop: equations for polV}---applied to the lattice of slender languages (see Section~\ref{SW:sec: density})--- to prove the decidability of the lattice  generated by the languages of the form $L_0a_1L_1 \cdots a_kL_k$ where the $L_i$ are either $A^*$ or of the form $u^*$ for some $u\in A^*$.
\end{example}

\subsection{Product operations and semidirect products}\label{SW:subsec:products}

We now consider products of the form $LaA^*$, where $L$ is a language and $a\in A$: $LaA^*$ is the language of all words with a prefix in $La$. Given a monoid $M$ accepting $L$, one can construct a monoid accepting $LaA^*$ using the operation of semidirect product.

In general, let $S$ and $T$ be monoids. A \emph{left action} of $T$ on $S$ is a mapping $\lambda\colon T\times S \to S$, written $(t,s) \mapsto t\cdot s$, such that for each $t$, the map $\lambda_t\colon s \mapsto t\cdot s$ is an endomorphism of $S$, and such that the map $t \mapsto \lambda_t$ is a morphism from $T$ to the monoid of endomorphisms of $S$. Once such an action $\lambda$ is given, the \emph{semidirect product}\index{product!semidirect} $S \ast_\lambda T$ (we usually write $S\ast T$) is the monoid of all pairs $(s,t) \in S \times T$, with product
\begin{displaymath}
(s,t)(s',t') = (s\ \lambda(t,s'),\ tt').
\end{displaymath}

\begin{lemma}\label{SW:claim: LaA*}
	If $\phi\colon A^* \to T$ accepts the language $L$, then $U_1^T \ast T$ accepts $LaA^*$.
\end{lemma}

\begin{proof}
We consider the action $\lambda$ of $T$ on $U_1^T$ given by $\lambda(t,(s_x)_{x\in T}) = (s'_x)_{x\in T}$, with $s'_x = s_{xt}$. Let then $\psi\colon A^* \to U_1^T\ast T$ be given by , for each $b\in A$,
\begin{align*}
	\psi(b) &= \left( (s_x^{(b)})_{x\in T},\ \phi(b)\right)\qquad\textrm{with} \\
	s_x^{(b)} &= \begin{cases} 0 & \textrm{if $x\in \phi(L)$ and $b = a$,}\\
	                                           1 & \textrm{otherwise.} \end{cases}
\end{align*}
Using the definition of the product in $U_1^T \ast T$, we find that
\begin{align*}
	\psi(a_1\cdots a_n) &= \left( (r_x)_{x\in T},\ \phi(a_1\cdots a_n)\right)\qquad\textrm{with}  \\
	 r_x &= s_x^{(a_1)}\ s_{x\phi(a_1)}^{(a_2)}\ \cdots\ s_{x\phi(a_1\cdots a_{n-1})}^{(a_n)} \\
	 &= \begin{cases} 0 & \textrm{if for some $1\le i\le n$, $x\phi(a_1\cdots a_{i-1}) \in \phi(L)$ and $a_i = a$,} \\
	                              1 & \textrm{otherwise.}\end{cases}
\end{align*}
In particular, we observe that $a_1\cdots a_n \in LaA^*$ if and only if $r_1 = 0$.
\end{proof}

\begin{remark}
Observe that the construction of the semidirect product $U_1^T\ast T$ given above does not use anything special about $U_1$, and thus can be applied to any pair of monoids $U$ and $T$.  This is called the \emph{wreath product}\index{product!wreath} $U \circ T$. The wreath product is closely related to the semidirect product, in the sense that first, it is, of course,  a semidirect product with $T$ of a member of the pseudovariety generated by $U$, and, second, every semidirect product $U\ast T$ embeds in $U \circ T$.  The wreath product, in a sense that can be made precise, captures the notion of series composition of automata \cite{Eilenberg:1976}.  As a consequence it is frequently used, exactly as in the proof of Lemma~\ref{SW:claim: LaA*} above to prove decomposition results.
\end{remark}

The operation of semidirect product is naturally extended to pseudovarieties: if $\psvV$ and $\psvW$ are pseudovarieties, we let $\psvV\ast\psvW$ be the pseudovariety generated by the semidirect products $S\ast T$ with $S\in \psvV$ and $T\in \psvW$. Then we have the following theorem.

\begin{theorem}\label{SW:thm: LaA*}
	Let $\calV$ be a variety of languages, and for each alphabet $A$, let $A^*\calW$ be the boolean algebra generated by the languages of $A^*\calV$ and the languages of the form $LaA^*$ with $L\in A^*\calV$. Then the class of languages $\calW$ is a variety and the corresponding pseudovariety of monoids is $\psvJ_1\ast\psvV$.
\end{theorem}

\begin{proof}
Since $U_1 \in \psvJ_1$, Lemma~\ref{SW:claim: LaA*} shows that every language in $A^*\calW$ is accepted by a monoid in $\psvJ_1\ast\psvV$. The proof of the converse is a particular case of the more general \emph{wreath product principle}\index{wreath product principle} (Straubing \cite{Straubing:1979a}). Let $\phi$ be a morphism $\phi\colon A^* \to S \ast T$ and for each $a\in A$, let $\phi(a) = (s_a,t_a)$. Let $\psi\colon A^* \to T$ be the morphism given by $\psi(a) = t_a$. Let also $B = T\times A$ and let $\sigma\colon A^* \to B^*$ be the map
\begin{displaymath}
\sigma(a_1\cdots a_n) = (1,a_1)\ (\psi(a_1),a_2)\ \cdots \ (\psi(a_1\cdots a_{n-1}),a_n).
\end{displaymath}
Note that $\sigma$ is a so-called sequential function \cite{Berstel:1979,Sakarovitch:2009}, not a morphism. We observe however that, if $\chi\colon B^* \to S$ is the morphism given by $\chi(t,a) = t\cdot s_a$, then
\begin{displaymath}
\phi(a_1\cdots a_n) = \left(\chi\sigma(a_1\cdots a_n), \ \psi(a_1\cdots a_n) \right).
\end{displaymath}
It follows that if $(s,t) \in S\ast T$, then $\phi\inv(s,t) =  \psi\inv(t) \cap \sigma\inv(\chi\inv(s))$. If $T\in \psvV$, then $\psi\inv(t) \in A^*\calV$. And if $S\in \psvJ_1$, then $\chi\inv(s)$ is a language in $B^*\calJ_1$, and hence a boolean combination of languages of the form $B^*(t,a)B^*$ ($(t,a) \in B$). Then $\sigma\inv(\chi\inv(s))$ is a boolean combination of languages of the form $\sigma\inv(B^*(t,a)B^*)$. Now $\sigma(a_1\cdots a_n) \in B^*(t,a)B^*$ if and only if, for some $1\le i \le n$, we have $(t,a) = (\psi(a_1\cdots a_{i-1}), a_i)$, that is, if and only if $a_1\cdots a_n \in \psi\inv(t)aA^*$. In particular, $\chi\inv(s)$ and $\phi\inv(s,t)$ are in $A^*\calW$, and so is any language accepted by $\phi$.
\end{proof}

\begin{remark}\label{SW:remark: complexity}
The semidirect product  is a powerful tool for decomposing pseudovarieties. The operation $\psvV\ast\psvW$ is associative on pseudovarieties and Krohn and Rhodes \cite{Krohn&Rhodes:1965} established\index{Krohn-Rhodes Theorem} that every finite monoid $M$ sits in an iterated product $\psvX_1\ast\cdots\ast \psvX_k$  where each $\psvX_i$ is either $\psvG$ or $\psvA$ (and the $\psvG$ and $\psvA$ factors alternate since $\psvG\ast\psvG = \psvG$ and $\psvA \ast\psvA = \psvA$). This gives rise to a famous open problem, the so-called complexity problem: given $M$, can we compute the minimum number of $\psvG$ factors in a  product of $\psvA$ and $\psvG$ containing $M$? 
\end{remark}

An analogous operation, the 2-sided semidirect product, can be used to handle the products of the form $KaL$ ( $K, L \subseteq A^*$). This time, we need to consider not only a left action of $T$ on $S$ (as for the semidirect product), but also a right action of $T$ on $S$, a map $\rho\colon S\times T \to S$, written $(s,t) \mapsto s \cdot t$, with the dual properties of a left action ($\rho_t\colon s \mapsto s\cdot t$ is an endomorphism of $S$ and $t \mapsto \rho_t$ is a morphism), and such that, for all $t,t'\in T$, $\lambda_t$ and $\rho_{t'}$ commute: $t\cdot (s\cdot t') = (t\cdot s)\cdot t'$. Then the \emph{2-sided semidirect product}\index{product!2-sided semidirect} $S \dast_{\lambda,\rho} T$ (written $S\dast T$) is the monoid of all pairs $(s,t) \in S \times T$, with product
\begin{displaymath}
(s,t)(s',t') = (\rho(s,t')\ \lambda(t,s'),\ tt').
\end{displaymath}
Again, the operation is extended to pseudovarieties, by letting $\psvV\dast\psvW$ be the pseudovariety generated by the products $S \dast T$ with $S\in \psvV$ and $T\in \psvW$. Then the following analogue of Theorem~\ref{SW:thm: LaA*} holds.

\begin{theorem}\label{SW:thm: KaL}
	Let $\calV$ be a variety of languages and for each alphabet $A$, let $A^*\calW$ is the boolean algebra generated by the languages of $A^*\calV$ and the languages of the form $KaL$ with $K,L\in A^*\calV$. Then the class $\calW$ is a variety and the corresponding pseudovariety of monoids is $\psvJ_1\dast\psvV$.
\end{theorem}

\begin{proof}
The first step of the proof consists in verifying that if $K$ and $L$ are accepted by a monoid in $T \in \psvV$, then $KaL$ is accepted by $U_1^{T\times T}\dast T$. (Note that if $K$ and $L$ are accepted by monoids $T_1$ and $T_2$, then they are both accepted by $T_1\times T_2$, so it is no restriction to assume that $K$ and $L$ are accepted by the same monoid.) This step is performed essentially as in Lemma~\ref{SW:claim: LaA*}, and the details are left to the reader.

The second step, to prove that if $\phi$ is a morphism $\phi\colon A^* \to S \dast T$ with $S\in \psvJ_1$ and $T\in \psvV$, then each $\phi\inv(s,t)$ is in $A^*\calW$. Here too, we use (a 2-sided version of) the wreath product principle \cite{Weil:1992}. For each $a\in A$, let $\phi(a) = (s_a,t_a)$. Let $\psi\colon A^* \to T$ be the morphism given by $\psi(a) = t_a$, let $B = T\times A\times T$ and let $\sigma\colon A^* \to B^*$ be the map
\begin{multline}
\sigma(a_1\cdots a_n) \\= (1,a_1,\psi(a_2\cdots a_n))\ (\psi(a_1),a_2,\psi(a_3\cdots a_n))\ \cdots \ (\psi(a_1\cdots a_{n-1}),a_n,1).\nonumber
\end{multline}
Then, if $\chi\colon B^* \to S$ is the morphism given by $\chi(t,a,t') = (t\cdot s_a)\cdot t'$, then
\begin{displaymath}
\phi(a_1\cdots a_n) = \left(\chi\sigma(a_1\cdots a_n), \ \psi(a_1\cdots a_n) \right).
\end{displaymath}
We conclude as in the proof of Theorem~\ref{SW:thm: LaA*}.
\end{proof}

\begin{remark}
In view of Sch\"utzenberger's theorem (Theorem~\ref{SW:thm: Schutz} above), one can use this result to show that the least pseudovariety closed under the operation $\psvV\mapsto \psvJ_1\dast\psvV$, is the pseudovariety $\psvA$ of aperiodic monoids.
\end{remark}

Semidirect product decomposition yields very difficult decision problems, such as the complexity problem briefly described in Remark~\ref{SW:remark: complexity}. Tilson showed that the consideration of certain categories offered a systematic tool for understanding semidirect (and 2-sided semidirect) product decompositions (\cite{Tilson:1987}, see also \cite{Straubing:1994}).  Almeida and Weil combined this category-theoretical approach with topological methods to provide sets of profinite identities describing many instances of semidirect products \cite{Almeida&Weil:1998}. As with Mal'cev products, these sets are usually infinite and do not offer immediate solutions to decidability problems, see \cite{Auinger&Steinberg:2003}.

For the products discussed in this section, \cite{Almeida&Weil:1998} gives the following descriptions.

\begin{proposition}
Let $\psvV$ be a pseudovariety of monoids. Then
$\psvJ_1\ast\psvV$ is defined by the set of profinite identities of the form $xy^2 = xy$ and $xyz = xzy$ for all $x,y,z \in \widehat{X^*}$ such that $\psvV$ satisfies $xy = xz = x$.

$\psvJ_1\dast\psvV$ is defined by the set of profinite identities of the form $xy^2x' = xyx'$ and $xyzx' = xzyx'$ for all $x,y,z,x' \in \widehat{X^*}$ such that $\psvV$ satisfies $xy = xz = x$ and $yx' = zx' = x'$.
\end{proposition}

In \cite{Almeida&Weil:1998}, this result is used to show the decidability of $\psvJ_1\ast\psvJ$ and $\psvJ_1\dast\psvJ$.

It is interesting also to note that 2-sided semidirect products and category-theoretical extensions of the notion of pseudovariety can be used to decompose unambiguous products, that is, to decompose the operation $\psvV \mapsto L\psvI\malcev\psvV$, see \cite{Pin&Straubing&Therien:1988}.

\section{Varieties in other algebraic frameworks}

The fundamental notions explored in this chapter---classes of algebras defined by identities, properties preserved under products and quotients, {\it etc.}---properly belong to the domain of universal algebra.  We have applied these ideas to finite monoids, ordered finite monoids, and stamps, but in fact they are applicable in a much wider variety of settings.  Here we will briefly discuss some of these extensions.

The study of varieties originates in the work of Birkhoff~\cite{Birkhoff:1935}, who showed that a family of algebras (defined in a very general sense) is closed under formation of subalgebras, quotients and products if and only if it is defined by a set of identities.  Such families of algebras are called {\it varieties} because of a loose analogy with the varieties of algebraic geometry defined by sets of polynomial equations.  Note that the classes of finite monoids that we have discussed are not varieties in this sense because they are not, of course, closed under infinite direct products, nor even finite quotients of infinite direct products, and consequently they cannot be defined by  sets of explicit identities (as opposed to profinite identities).

Efforts to adapt Birkhoff's Theorem to finite algebras include work of Eilenberg and Sch\"utzenberger~\cite{Eilenberg&Schutzenberger:1976}, and of Baldwin and Berman~\cite{Baldwin&Berman:1976}, who both showed that pseudovarieties are indeed defined by sets of identities, in the sense that an algebra belongs to a pseudovariety if and only if it satisfies all but finitely many identities of the set.   A different treatment, and the one that we have followed here, based on identities in free profinite algebras, was given by Reiterman, who proved the second part of Theorem 2.15 in the setting of arbitrary finite algebras~\cite{Reiterman:1982} (see also Banaschewski \cite{Banaschewski:1983}).

The first part of Theorem 2.15, characterizing the language classes corresponding to pseudovarieties of finite monoids, is from Eilenberg~\cite{Eilenberg:1976}.  A generalization applicable to pseudovarieties of single-sorted finite algebras is given by Almeida~\cite{Almeida:1994}.

The ordered monoids considered in this chapter are not, strictly speaking, algebras, but rather instances of finite {\it ${\cal L}$-structures,} which are algebras together with a set of relations compatible with the operations in the algebra.  Pin and Weil~\cite{Pin&Weil:1996} prove an analogue of Reiterman's Theorem for such structures.  In this setting the profinite identities are replaced by profinite relational identities. The profinite ordered identities discussed in this chapter are a particular instance.

Variety theories of the kind described here have also been successfully extended to a number of many-sorted algebras that arise in the domain of automata theory, and which we briefly describe:

Wilke~\cite{Wilke:1993} and Perrin and Pin~\cite{Perrin&Pin:2004} consider regular languages of infinite words.  Here the corresponding algebraic objects are two-sorted algebras called {\it $\omega$-semigroups.} These are pairs $(S_f,S_{\omega})$, where $S_f$ is a semigroup, and where there are additional operations $S_f\times S_{\omega}\to S_{\omega}$ and $S_f \to S_{\omega}$.  Here the free object (analogous to the free monoid in the case of pseudovarieties of finite monoids) is the pair $(A^+,A^{\omega})$ of finite and infinite words over $A$. The three operations correspond to ordinary concatenation of finite words, concatenation of a finite word and an infinite word to obtain an infinite word, and taking the infinite power of a finite word to obtain an infinite word.  
  
  \'Esik and Weil~\cite{Esik&Weil:2005,Esik&Weil:2010} describe a theory of varieties for regular languages of {\it ranked trees}.  These are finite trees in which the nodes are labeled by letters of a finite alphabet $\Sigma$ that is the disjoint union of subalphabets $\Sigma_0,\ldots,\Sigma_n$, where the label of a node with $k$ children belongs to $\Sigma_k$.  In particular, the number of children of any node in such a tree is bounded above by $n$.  The corresponding algebraic objects are called {\it finitary preclones.} These are sequences of finite sets $S_0,S_1,\ldots$.  The operation takes an element $f$ of $S_k$, and a sequence $g=(g_1,\ldots,g_k)$, where $g_i\in S_{m_i}$, and yields an element $f\cdot g$ of $S_m$, where $m=m_1+\cdots + m_k$.  The free object is the sequence $(\Sigma M_0,\Sigma M_1,\ldots)$, where $\Sigma M_k$ consists of {\it $k$-ary ranked trees}: these are ranked trees in which $k$ of the leaves, reading in left-to-right order, have been replaced by the variable symbol $v_1,\ldots, v_k$.  In this free preclone, the operation $f\cdot (g_1,\ldots,g_k)$ is that of replacing the $k$ variables in $f$ by the trees $g_1,\ldots, g_k$ to obtain an $m$-ary ranked tree.  
  
The theory can be extended as well to regular languages of finite unranked forests, in which there is no bound on the degree of branching of the nodes (e.g., Bojanczyk and Walukiewicz \cite{Bojanczyk&Walukiewicz:2008}, Bojanczyk, Straubing and Walukiewicz \cite{Bojanczyk&Straubing&Walukiewicz:2012}).  Here the corresponding algebraic objects are called {\it forest algebras}.  These are pairs $(H,V)$ of monoids where $V$ acts on $H$.  The letters $H$ and $V$ stand for `horizontal' and `vertical':  The free object is the pair $(H_A,V_A)$ where $H_A$ consists of forests labeled by letters of $A$, and $V_A$ consists of {\it contexts}: forests in which the letter at one leaf has been deleted and replaced by a single variable.  The product in $H_A$ is simply concatenation of forests to obtain larger forests; the product in $V_A$ is substitution of one context for the variable in another context; and the action of $V_A$ on $H_A$ is substitution of a forest for the variable in a context so as to obtain a larger forest.

For further details on this algebraic approach of the theory of regular tree
languages, we refer the reader to Chapter~\ref{Chap22} in this Handbook.


\bibliographystyle{abbrv}
\addcontentsline{toc}{section}{References}
\begin{footnotesize}
\bibliography{abbrevs,StraubingWeilBiblio}
\end{footnotesize}

\newpage{\pagestyle{empty}\cleardoublepage}

\end{document}